\documentclass[10pt]{article}
\usepackage[T1]{fontenc}
\usepackage{amsmath,amsfonts,amsthm,mathrsfs,amssymb}
\usepackage{cite}
\usepackage{multirow}
\usepackage{graphicx,float}
\usepackage{subfigure}
\usepackage{placeins}
\usepackage{color}
\usepackage{indentfirst}
\numberwithin{equation}{section}
\newcommand{\R}{{\mathbb R}}

\newcommand{\be}{\begin{eqnarray}}
\newcommand{\ben}{\begin{eqnarray*}}
\newcommand{\en}{\end{eqnarray}}
\newcommand{\enn}{\end{eqnarray*}}

\newcommand{\bs}{\boldsymbol}
\newcommand{\pa}{\partial}

\newcommand{\curl}{{\rm curl\,}}

\newcommand{\divv}{{\rm div\,}}

\newtheorem{theorem}{Theorem}[section]
\newtheorem{lemma}[theorem]{Lemma}

\newtheorem{remark}[theorem]{Remark}

\definecolor{rot}{rgb}{1.000,0.000,0.000}
\definecolor{rot1}{rgb}{0.000,0.000,0.000}
\begin{document}
\renewcommand{\theequation}{\arabic{section}.\arabic{equation}}
\begin{titlepage}
\title{Helmholtz decomposition based windowed Green function methods for elastic scattering problems on a half-space}

\author{
Tao Yin\thanks{LSEC, Institute of Computational Mathematics and Scientific/Engineering Computing, Academy of Mathematics and Systems Science, Chinese Academy of Sciences, Beijing 100190, China. Email:{\tt yintao@lsec.cc.ac.cn}},
Lu Zhang\thanks{School of Mathematical Sciences, Zhejiang University, Hangzhou 310027, China. Email:{\tt 0623216@zju.edu.cn}},
Weiying Zheng\thanks{LSEC, NCMIS, Institute of Computational Mathematics and Scientific/Engineering Computing, Academy of Mathematics and System Sciences, Chinese Academy of Sciences, Beijing, 100190, China. School of Mathematical Science, University of Chinese Academy of Sciences, Beijing 100049, China. Email:{\tt zwy@lsec.cc.ac.cn}},
Xiaopeng Zhu\thanks{School of Mathematical Sciences, Zhejiang University, Hangzhou 310027, China. Email:{\tt 0623559@zju.edu.cn}}}
\date{}
\end{titlepage}
\maketitle

\begin{abstract}
This paper proposes a new Helmholtz decomposition based windowed Green function (HD-WGF) method for solving the time-harmonic elastic scattering problems on a half-space with Dirichlet boundary conditions in both 2D and 3D. The Helmholtz decomposition is applied to separate the pressure and shear waves, which satisfy the Helmholtz and Helmholtz/Maxwell equations, respectively, and the corresponding boundary integral equations of type $(\mathbb{I}+\mathbb{T})\bs\phi=\bs f$, that couple these two waves on the unbounded surface, are derived based on the free-space fundamental solution of Helmholtz equation. This approach avoids the treatment of the complex elastic displacement tensor and traction operator that involved in the classical integral equation method for elastic problems. Then a smooth ``slow-rise'' windowing function is introduced to truncate the  boundary integral equations and a ``correction'' strategy is proposed to ensure the uniformly fast convergence for all incident angles of plane incidence. Numerical experiments for both two and three dimensional problems are presented to demonstrate the accuracy and efficiency of the proposed method.

\end{abstract}

\section{Introduction}
\label{sec:1}

The problems of scattering of acoustic, electromagnetic, and elastic waves by unbounded rough surface (or called {\it layered-medium problems}, {\it half-space problems} sometimes) are of significant importance in many application fields in science and engineering and they have been attracted the interest of researchers in both engineering and mathematical circles for many years. This work focuses on developing new efficient high-order integral equation solvers for the elastic scattering problems on a half-space~\cite{A73,S58}. Compared to the problems with bounded obstacles in free space, the study of the half-space problems generally suffer more challenges in the issues of {well-posedness}~\cite{A01,A02,EH12,EH15,HLQZ15} and numerical approximation~\cite{BY21,CB13,CBS08}. In contrast to volumetric discretization methods, like the finite eleme/{difference} methods, the boundary integral equation (BIE) method~\cite{HW08} only requires discretization of regions of lower dimensionality and the radiation condition at infinity can be automatically enforced. In addition, together with adequate acceleration techniques~\cite{BK01,CBS08} for the associated matrix-vector products and preconditioning~\cite{BY20} to improve the spectral properties of the linear system, the BIE method can provide fast, high-order solvers even for problems of high frequency.

The method of using layer Green function (LGF)~\cite{CB14} is one of the most popular integral equation approach for the half-space scattering problems. This method brings convenience for the study of wellposedness~\cite{HYZ18} and it automatically enforces the relevant boundary conditions on the unbounded flat surface and hence, can reduce the original problems to integral equations only on the defects. However, evaluation of the LGF requires computation of challenging Fourier integrals containing highly-oscillatory integrands over infinite integration intervals. In particular, the LGF for the elastic problems is much more complex~\cite{A00,CZ19,DMN11} than that for acoustic problems. Another efficient integral equation approach to treat the half-space scattering problems is the method of using free-space Green function (FGF) \cite{DM97,DM98} whose evaluation is much lower, by order of magnitude, than the LGF evaluation cost. But the BIEs based on FGF are posed on the complete unbounded interface which means that, for computation purpose, appropriate domain-truncation strategy must be introduced and therefore, problems with regard to selection of suitable truncation radii {arise.} It is suggested in~\cite{CBS08,CB13} that a truncation radius equal to three to five times the radius of the surface irregularity yields acceptable accuracy for the elastic half-space problems with normal-incidence. However, as illustrated in \cite{BY21}, the numerical accuracy using this principle will deteriorate quickly as the {incident} angle approaches grazing.

The windowed Green function (WGF) method~\cite{P16,LGO18} has been proved to be an efficient truncation and modification technique to ensure the uniform fast convergence over all incident angles as the size of truncation domain grows for the method of using FGF~\cite{BLPT16,BP17} and this method has also been extended to the problems of scattering by periodic structures~\cite{BD14,BSTV16,M07}, nonuniform waveguide problems~\cite{BGP17} and long-range volumetric propagation~\cite{CBA09}. Utilizing smooth operator windowing based on a ``slow-rise'' windowing function, efficient high-order singular-integration methods based on the Chebyshev-based rectangular-polar discretization methodology~\cite{BG20,BY20} and equivalent regularized formulations of the strongly-singular integral operators, the corresponding WGF method for solving the elastic half-space problems was first developed in \cite{BY21}. It is also meaningful to mention that recently, perfectly-matched-layer (PML) based integral equation solvers~\cite{LLQ18,LXYZ23} have been developed for the layered-medium scattering problem. But the application of this approach to the elastic layered-{medium} problems has not been developed yet which will be left for future work and in particular, the convergence of the PML truncation remains open.

In contrast to \cite{BY21}, this paper proposes a new Helmholtz decomposition based windowed Green function (HD-WGF) method for solving the elastic half-space problems. The Helmholtz decomposition allows us to split the displacement of the elastic wave field into the compressional wave and the shear wave which satisfy the Helmholtz and Helmholtz/Maxwell equations, respectively. As a result, the elastic scattering problem can be converted equivalently into a coupled boundary value problem of the Helmholtz and Maxwell equations for the potentials. This technique has been introduced into the study of the BIE methods for the problems of elastic scattering by bounded obstacles in~\cite{DLL21,DLL22,LL19} and reduces greatly the complexity for the computation of the elastic scattering problem since, compared with the classical BIE methods, it avoids the treatment of the full elastic fundamental tensor and the traction operator. This work presents the first attempt to apply the Helmholtz decomposition to the more difficult half-space problems. Unlike the procedure in~\cite{DLL21,DLL22,LL19} using indirect layer potential representations of the solutions, the direct solution representations resulting from Green's formula should be adopted here so that the modification idea of the WGF method can be incorporated into the windowed integral equations, and then, uniform fast convergence over all incident angles can be achieved. Analogous to~\cite{BY21}, the Chebyshev-based rectangular-polar discretization methodology is utilized for the numerical implementation of the HD-WGF method which, as a result of the Helmholtz decomposition, only requires the evaluation of several weakly-singular integral operators in terms of free-space fundamental solution to the Helmholtz equations with compressional and shear wave numbers as well as some surface differential operators.

This paper is organized as follows. Section~\ref{sec:2} introduces the considered elastic half-space problems. Section~\ref{sec:3} presents the Helmholtz decompositions and the corresponding BIEs in both 2D and 3D based on the free-space Green function. The HD-WGF method is proposed in Section~\ref{sec:4}, including the preliminary windowed integral formulations and ``corrected'' ones which are uniformly accurate for all incident angles, up to grazing. Numerical examples in 2D and 3D, finally, are presented in Section~\ref{sec:5} to demonstrate the accuracy and efficiency of the overall proposed approach.

\section{Elastic scattering problems}
\label{sec:2}

\begin{figure}[ht]
\centering
\includegraphics[scale=0.3]{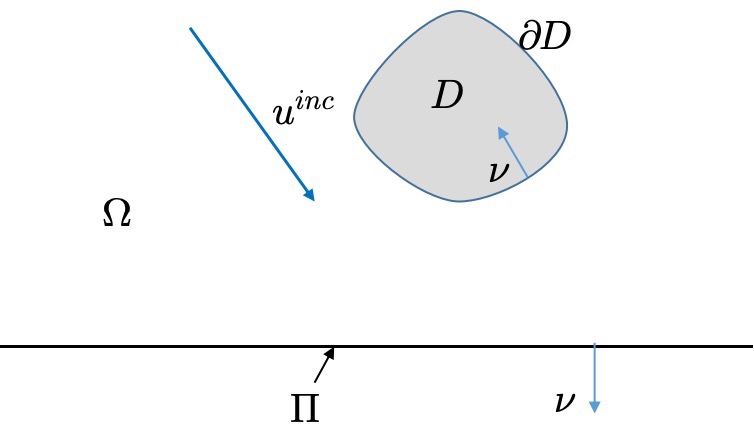}
\caption{Problem of scattering by a locally perturbed elastic
  half-space $\Omega\subseteq \mathbb{R}^d$ ($d=2,3$).}
\label{model}
\end{figure}

As shown in Figure~\ref{model}, let $D\in\R^d_+,d=2,3$ be a bounded obstacle with Lipschitz boundary $\pa D$ and denote by $\Omega\in\R^d_+$ the unbounded connected open set exterior to $D$ such that
\ben
U_{f_+}\subset\Omega\subset U_{f_-},\quad
U_{f_\pm}:=\{{\bs x}=(x_1,\dots,x_d)^\top\in\R^d: x_d>f_\pm\}
\enn
for certain constants $f_-<f_+$. For simplicity, the unbounded rough surface $\Gamma:=\pa\Omega$ is assumed to encompasse the
flat surface $\Pi:=\{{\bs x}\in\R^d: x_d=0\}$ and $\pa D$, but the  approach shown in this work can be extended to the more general case that the unbounded rough surface $\Gamma$ consists both a local perturbation of the flat surface $\Pi$ and/or the boundary of bounded impenetrable obstacles in $U_0$, see Remark~\ref{localdefect2D}. Let  $\nu=(\nu_1,\nu_2)^\top$ be the outward unit normal to $\Gamma$ and $\partial_\nu:=\nu\cdot\mbox{grad}$ be the normal derivative. Suppose that the $\Omega$ is full filled with a linear isotropic and homogeneous elastic medium characterized by the Lam\'e constants $\lambda,\mu$ ($\mu>0$, $d\lambda+2\mu>0$) and the mass density $\rho>0$. Given an incident field ${\bs u}^\mathrm{inc}$, we consider the following time-harmonic elastic scattering problems on a half-space: find the scattered displacement field ${\bs u}^\mathrm{sca}$ satisfying the Navier equation
\be
\label{navier}
\Delta^*{\bs u}^{\mathrm{sca}}+\rho\omega^2{\bs u}^{\mathrm{sca}}=0 \quad\mbox{in}\quad\Omega,
\en
the Dirichlet boundary condition
\be
\label{Dirichlet}
{\bs u}^{\mathrm{sca}}=-{\bs u}^\mathrm{inc}\quad\mbox{on}\quad\Gamma,
\en
and the so-called upward propagating radiation condition (UPRC) at infinity~\cite{EH12,CGK02}. Here, the Lam\'e operator $\Delta^{*}$ is defined as
\ben
\label{LameOper}
\Delta^* := \mu\Delta + (\lambda +
\mu)\nabla\nabla\cdot,
\enn
and $\omega$ is the angular frequency. Moreover, for the considered linear isotropic medium, the pressure and shear wave numbers $k_p,k_s$ are given by
\ben
k_s := \omega\sqrt{\rho/\mu},\quad k_p :=
\omega \sqrt{\rho/(\lambda + 2\mu)}.
\enn

In this work, the incident field ${\bs u}^\mathrm{inc}$ is specified as a coupling of the plane pressure wave ${\bs u}_p^\mathrm{inc}$ and shear wave ${\bs u}_s^\mathrm{inc}$ with different incident angles, i.e.,
\ben
{\bs u}^\mathrm{inc}= c_p{\bs u}_p^\mathrm{inc}+c_s{\bs u}_s^\mathrm{inc},
\enn
where $c_\xi,\xi=p,s$ are constants.
\begin{itemize}
\item {\bf Two-dimensional case.} Let $\theta_p^{\rm inc}$ and $\theta_s^{\rm inc}$ be the incident angles of the pressure and shear waves, respectively, satisfying $|\theta_\xi^{\rm inc}|<\pi/2, \xi=p,s$. Denote
\ben
\alpha_\xi=k_\xi\sin\theta_\xi^{\rm inc},\quad \xi=p,s,
\enn
and correspondingly,
\ben
\beta_{\xi_1,\xi_2}=\begin{cases}
\sqrt{k_{\xi_1}^2-\alpha_{\xi_2}^2}, & |\alpha_{\xi_2}|\le k_{\xi_1},\cr
i\sqrt{\alpha_{\xi_2}^2 -k_{\xi_1}^2} , & |\alpha_{\xi_2}|> k_{\xi_1},
\end{cases}
 \quad \xi_1,\xi_2=p,s.
\enn
The incident fields ${\bs u}_\xi^\mathrm{inc},\xi=p,s$ in two dimensions are given by
\be
\label{2Dinc}
{\bs u}_p^\mathrm{inc}(\bs x)= \frac{i}{k_p}\begin{pmatrix}
\alpha_p \\
-\beta_{p,p}
\end{pmatrix}e^{i(\alpha_px_1-\beta_{p,p}x_2)},\quad  {\bs u}_s^\mathrm{inc}(\bs x)=-\frac{i}{k_s}\begin{pmatrix}
\beta_{s,s} \\
\alpha_s
\end{pmatrix}e^{i(\alpha_sx_1-\beta_{s,s}x_2)}.
\en

\item {\bf Three-dimensional case.} Let $(\theta_p^{\rm inc},\varphi_p^{\rm inc})$ and $(\theta_s^{\rm inc},\varphi_s^{\rm inc})$ be the incident angle pairs of the pressure and shear waves, respectively, satisfying $|\theta_\xi^{\rm inc}|<\pi/2,  \varphi_\xi^{\rm inc}\in[0,2\pi), \xi=p,s$. Denote $\widetilde{\bs x}=(x_1,x_2)^\top$ and
\ben
{{\bs \alpha}_\xi=k_\xi(\sin\theta_\xi^{\rm inc}\cos\varphi_\xi^{\rm inc}, \sin\theta_\xi^{\rm inc}\sin\varphi_\xi^{\rm inc})},\quad \xi=p,s,
\enn
and correspondingly,
\ben
\beta_{\xi_1,\xi_2}=\begin{cases}
\sqrt{k_{\xi_1}^2-|\bs\alpha_{\xi_2}|^2}, & |\bs\alpha_{\xi_2}|\le k_{\xi_1},\cr
i\sqrt{|\bs\alpha_{\xi_2}|^2 -k_{\xi_1}^2} , & |\bs\alpha_{\xi_2}|> k_{\xi_1},
\end{cases}
 \quad \xi_1,\xi_2=p,s.
\enn
Then given a polarization vector $\bs d_s\in\mathbb{S}^2=\{\bs x\in{\R^3}: |\bs x|=1\}$ satisfying $\bs d_s\cdot (\bs \alpha_s^\top, -\beta_{s,s})^\top=0$, the incident fields ${\bs u}_\xi^\mathrm{inc},\xi=p,s$ in three dimensions are given by
\be
\label{3Dinc}
{\bs u}_p^\mathrm{inc}(\bs x)= \frac{i}{k_p}\begin{pmatrix}
\bs\alpha_p^\top \\
-\beta_{p,p}
\end{pmatrix}e^{i(\bs\alpha_p\cdot \widetilde{\bs x}-\beta_{p,p}x_3)},\quad {{\bs u}_s^\mathrm{inc}(\bs x)= \frac{i}{k_s}\begin{pmatrix}
\bs \alpha_s^\top\\
-\beta_{s,s}
\end{pmatrix}\times \bs d_se^{i(\bs\alpha_s\cdot \widetilde{\bs x}-\beta_{s,s}x_3)}.}
\en
\end{itemize}



\section{Helmholtz decompositions and boundary integral equations}
\label{sec:3}

In this section, instead of solving the elastic displacement field, we rewrite the elastic scattering problems by means of Helmholtz decomposition which decomposes the elastic field into a compressional part and a shear part, and then derive the corresponding boundary integral equations based on the free-space Green function. In the following, denote by $G_k(\cdot,\cdot)$ the free-space Green function of the Helmholtz equation with wavenumber $k>0$ and it is given by
\ben
G_{k}(x,y)=\begin{cases}
\frac{i}{4}H_0^{(1)}(k|x-y|), & d=2,\cr
\frac{e^{ik|x-y|}}{4\pi|x-y|}, & d=3,
\end{cases}\quad x\ne y,\quad x,y\in\R^d.
\enn
For an arbitrary surface $\gamma$ and $\xi=p,s$, let $\mathcal{S}_\xi^\gamma$ and $\mathcal{D}_\xi^\gamma$ denote the single- and double-layer potentials, respectively, that are defined as
\ben
\mathcal{S}_\xi^\gamma[\phi](x)&=&\int_\gamma G_{k_\xi}(x,y)\phi(y)ds_y,\quad x\notin\gamma,\\
\mathcal{D}_\xi^\gamma[\psi](x)&=&\int_\gamma
\partial_{\nu_y}G_{k_\xi}(x,y)\psi(y)ds_y,\quad x\notin\gamma.
\enn
Accordingly,
\be
\label{slo}
V_\xi^\gamma[\phi](x)&=& \int_\gamma G_{k_\xi}(x,y)\phi(y)ds_y,\quad x\in\gamma,\\
\label{dlo}
K_\xi^\gamma[\psi](x)&=& \int_\gamma \partial_{\nu_y}G_{k_\xi}(x,y)\psi(y)ds_y,\quad x\in\gamma
\en
denote the single- and double-layer integral operators on $\gamma$, respectively.

\begin{remark}
The windowed Green function method proposed in \cite{BY21} solves the elastic scattering problems on a half space based on the elastic fundamental tensor which is given by
\ben
E(x,y)=\frac{1}{\mu}G_{k_s}(x,y)I+\frac{1}{\rho\omega^2}
\nabla_x\nabla_x^\top \left[G_{k_s}(x,y)-G_{k_p}(x,y)\right].
\enn
For the Dirichlet problem, the double-layer integral operator
\be
\label{tdlo}
K[\psi](x)&=&\int_\Gamma T(\pa_x,\nu_x)E(x,y)\psi(y)ds_y,\quad
x\in\Gamma, \en
is involved in the boundary integral equation whereas $T(\pa,\nu)$ denotes the traction operator defined as
\be
\label{stress}
T(\pa,\nu)u:=2 \mu \, \partial_{\nu} u + \lambda \, \nu \, \divv u+\mu
\nu\times \curl u.
\en
The new approach proposed in this work can avoid the treatment of the complicated elastic displacement tensor and the complicated kernel in elastic double-layer integral operator $K$. No use of the traction operator is required.
\end{remark}

\subsection{Two-dimensional case}
\label{sec:3.1}

Denote $\curl v=\pa_{x_1}v_2-\pa_{x_2}v_1$ for a vector function $v=(v_1,v_2)^\top$ and $\overrightarrow{\mathrm{curl}}\,w=(\pa_{x_2}w,-\pa_{x_1}w)^\top$ for a scale function $w$.  Correspondingly, we denote by $\tau=(-\nu_2,\nu_1)^\top$ and $\pa_\tau:=\tau\cdot\mbox{grad}$ the unit tangential vectors on $\Gamma$ and the tangential derivative, respectively. It is known that the solution $\bs u^\mathrm{sca}$ of the Navier equation (\ref{navier}) in two dimensions admits the Helmholtz decomposition
\be
\label{2DHDsca}
\bs u^\mathrm{sca}=\nabla \phi_p^\mathrm{sca}+ \overrightarrow{\mathrm{curl}}\,\phi_s^\mathrm{sca},
\en
where $\phi_p^\mathrm{sca},\psi_s^\mathrm{sca}$ are two scalar functions satisfying the Helmholtz equations
\be
\label{2Dhelmholtz}
\Delta\phi_\xi^\mathrm{sca}+k_\xi^2\phi_\xi^\mathrm{sca}=0,\quad \xi=p,s.
\en
Note that the incident field $\bs u^\mathrm{inc}$ also admits a form of Helmholtz equations reading
\be
\label{2DHDinc}
\bs u^\mathrm{inc}=\nabla \phi_p^\mathrm{inc}+ \overrightarrow{\mathrm{curl}}\,\phi_s^\mathrm{inc},
\en
where
\ben
\phi_p^\mathrm{inc}=\frac{c_p}{k_p}e^{i(\alpha_px_1-\beta_{p,p}x_2)},\quad  \phi_s^\mathrm{inc}=\frac{c_s}{k_s}e^{i(\alpha_sx_1-\beta_{s,s}x_2)},
\enn
and they satisfy the Helmholtz equations
\ben
\Delta\phi_\xi^\mathrm{inc}+k_\xi^2\phi_\xi^\mathrm{inc}=0,\quad \xi=p,s.
\enn
Denote $\phi_\xi^\mathrm{tot}=\phi_\xi^\mathrm{sca}+\phi_\xi^\mathrm{inc}$, then the original two-dimensional elastic scattering problem can be reduced to the following coupled Helmholtz system
\be
\label{2DHS1}
\begin{cases}
\Delta\phi_p^\mathrm{sca}+k_p^2\phi_p^\mathrm{sca}=0,\quad \Delta\phi_s^\mathrm{sca}+k_s^2\phi_s^\mathrm{sca}=0\quad \mbox{in}\quad \Omega,\\
\pa_\nu\phi_p^\mathrm{tot}+\pa_\tau\phi_s^\mathrm{tot}=0,\quad \pa_\tau\phi_p^\mathrm{tot}-\pa_\nu\phi_s^\mathrm{tot}=0\quad \mbox{on}\quad \Gamma,
\end{cases}
\en
and $\phi_p^\mathrm{sca},\phi_s^\mathrm{sca}$ satisfy the half-plane UPRC at infinity.

\begin{lemma}
The solutions $\phi_\xi^\mathrm{tot},\xi=p,s$ on $\Gamma$ in two dimensions satisfy the boundary integral equations
\be
\label{2DBIE}
\left(\frac{\mathbb{I}}{2}+\mathbb{T}_\Gamma\right)\begin{bmatrix}
\phi_p^\mathrm{tot}\\
\phi_s^\mathrm{tot}
\end{bmatrix}=\begin{bmatrix}
\phi_p^\mathrm{inc}\\
\phi_s^\mathrm{inc}
\end{bmatrix} \quad\mbox{on}\quad\Gamma \quad\mbox{with}\quad \mathbb{T}_\Gamma=\begin{bmatrix}
K_p^\Gamma & V_p^\Gamma\pa_\tau \\
-V_s^\Gamma\pa_\tau & K_s^\Gamma
\end{bmatrix}.
\en
\end{lemma}
\begin{proof}
For each $\xi=p,s$ and following the analysis in~\cite{DM97}, the scattered fields $\phi_\xi^\mathrm{sca}$ admit the representations
\be
\label{2Dscawaverepre}
\phi_\xi^\mathrm{sca}(x)= \mathcal{S}_\xi^\Gamma[\partial_\nu\phi_\xi^\mathrm{sca}](x)-\mathcal{D}_\xi^\Gamma[\phi_\xi^\mathrm{sca}](x),\quad x\in\Omega.
\en
When $|\theta_\xi^\mathrm{inc}|<\pi/2, \xi=p,s$, it can be verified that the incident field $\phi_\xi^\mathrm{inc}$ satisfies
\be
\label{2Dincwaverepre}
0=\mathcal{S}_\xi^\Gamma[\partial_\nu\phi_\xi^\mathrm{inc}](x)-\mathcal{D}_\xi^\Gamma[\phi_\xi^\mathrm{inc}](x),\quad x\in\Omega.
\en
Then combining the formulations (\ref{2Dscawaverepre}) and (\ref{2Dincwaverepre}) and applying the boundary conditions in (\ref{2DHS1}) yield the solution representations
\be
\label{2DscawaveP}
\phi_p^\mathrm{sca}(x) &=& -\mathcal{S}_p^\Gamma[\partial_\tau\phi_s^\mathrm{tot}](x)-\mathcal{D}_p^\Gamma[\phi_p^\mathrm{tot}](x),\quad x\in\Omega,\\
\label{2DscawaveS}
\phi_s^\mathrm{sca}(x) &=& \mathcal{S}_s^\Gamma[\partial_\tau\phi_p^\mathrm{tot}](x)-\mathcal{D}_s^\Gamma[\phi_s^\mathrm{tot}](x),\quad x\in\Omega.
\en
Taking the limit as $x\rightarrow\Gamma$ in (\ref{2DscawaveP}) and (\ref{2DscawaveS}) and using the well-known jump relations associated with the acoustic layer potentials~\cite{CK13}, we obtain the BIEs
\be
\label{2DBIE1}
\frac{1}{2}\phi_p^\mathrm{tot}+K_p^\Gamma[\phi_p^\mathrm{tot}]+V_p^\Gamma[\pa_\tau\phi_s^\mathrm{tot}] &=& \phi_p^\mathrm{inc} \quad\mbox{on}\quad\Gamma,\\
\label{2DBIE2}
\frac{1}{2}\phi_s^\mathrm{tot}+K_s^\Gamma[\phi_s^\mathrm{tot}]-V_s^\Gamma[\pa_\tau\phi_p^\mathrm{tot}] &=& \phi_s^\mathrm{inc} \quad\mbox{on}\quad\Gamma,
\en
which completes the proof.
\end{proof}

\subsection{Three-dimensional case}
\label{sec:3.2}

To derive the boundary integral equations in three dimensions, for an arbitrary surface $\gamma$, we denote $\nabla_\gamma$ and $\mathrm{div}_\gamma$ the tangential gradient and surface divergence on $\gamma$, respectively, which are defined for a scalar function $u$ and a vector function $\bs v$ by the following equalities
\ben
\nabla_\gamma u=\nabla u-\nu\partial_\nu u,\quad \mathrm{div}_\gamma\bs v=\mathrm{div}\bs v-\nu\cdot\partial_\nu \bs v.
\enn
Denote $\bs\curl_\gamma=\nu\times \nabla_\gamma$ the vector surface curl on $\gamma$.

For the three-dimensional case, it is known that the solution $\bs u^\mathrm{sca}$ of the Navier equation (\ref{navier}) admits the Helmholtz decomposition
\be
\label{3DHDsca}
\bs u^\mathrm{sca}=\nabla \phi_p^\mathrm{sca}+ \nabla\times\bs \phi_s^\mathrm{sca},
\en
where $\phi_p^\mathrm{sca}$ is a scalar function satisfying the Helmholtz equation
\be
\label{3Dhelmholtz}
\Delta\phi_p^\mathrm{sca}+k_p^2\phi_p^\mathrm{sca}=0,
\en
and $\bs\phi_s^\mathrm{sca}$ is a vector function satisfying $\nabla\cdot\bs\phi_s^\mathrm{sca}=0$ and the Maxwell equation
\be
\label{3Dmaxwell}
\nabla\times\nabla\times\bs\phi_s^\mathrm{sca}-k_s^2\bs\phi_s^\mathrm{sca}=0.
\en
The considered incident field $\bs u^\mathrm{inc}$ in three dimensions admits a form of Helmholtz decomposition reading
\be
\label{3DHDinc}
\bs u^\mathrm{inc}=\nabla \phi_p^\mathrm{inc}+ \nabla\times\bs \phi_s^\mathrm{inc},
\en
where
\ben
\phi_p^\mathrm{inc}=\frac{c_p}{k_p}e^{i(\bs\alpha_p\cdot\widetilde{x}-\beta_{p,p}x_3)},\quad  \bs\phi_s^\mathrm{inc}=\frac{c_s}{k_s}\bs d_se^{i(\bs\alpha_s\cdot\widetilde{x}-\beta_{s,s}x_3)},
\enn
and they satisfy the Helmholtz equation (\ref{3Dhelmholtz}) and Maxwell equation (\ref{3Dmaxwell}), respectively, and additionally, $\nabla\cdot \bs\phi_s^\mathrm{inc}=0$ since $\bs d_s\cdot (\bs \alpha_s; -\beta_{s,s})=0$.
Denote $\phi_p^\mathrm{tot}=\phi_p^\mathrm{sca}+\phi_p^\mathrm{inc}$, $\bs\phi_s^\mathrm{tot}=\bs\phi_s^\mathrm{sca}+\bs\phi_s^\mathrm{inc}$, then we can obtain the following coupled Helmholtz system
\be
\label{3DHS1}
\begin{cases}
\Delta\phi_p^\mathrm{sca}+k_p^2\phi_p^\mathrm{sca}=0,\quad \nabla\times\nabla\times\bs\phi_s^\mathrm{sca}-k_s^2\bs\phi_s^\mathrm{sca}=0\quad \mbox{in}\quad \Omega,\\
\pa_\nu\phi_p^\mathrm{tot}+\nu\cdot\nabla\times\bs\phi_s^\mathrm{tot}=0,\quad \bs\curl_\Gamma\phi_p^\mathrm{tot}+\nu\times\nabla\times\bs\phi_s^\mathrm{tot}=0\quad \mbox{on}\quad \Gamma.
\end{cases}
\en

\begin{lemma}
The solutions $\phi_p^\mathrm{tot}$ and $\bs\phi_s^\mathrm{tot}$ on $\Gamma$ in three dimensions satisfy the boundary integral equations
\be
\label{3DBIE}
\left(\frac{\mathbb{I}}{2}+\mathbb{T}_\Gamma\right)\begin{bmatrix}
\phi_p^\mathrm{tot}\\
\bs\phi_s^\mathrm{tot}\times\nu
\end{bmatrix}=\begin{bmatrix}
\phi_p^\mathrm{inc}\\
\bs\phi_s^\mathrm{inc}\times\nu
\end{bmatrix}\quad\mbox{on}\quad\Gamma \quad\mbox{with}\quad\ \mathbb{T}_\Gamma=\begin{bmatrix}
K_p^\Gamma & V_p^\Gamma\mathrm{div}_\Gamma \\
\nu\times V_s^\Gamma\bs\curl_\Gamma & M_s^\Gamma
\end{bmatrix},
\en
where the integral operator $M_s^\Gamma$, called the magnetic dipole operator, is given by
\ben
M^\Gamma_s[\bs \phi](x)=\int_\Gamma \left[\nabla_xG_{k_s}(x,y)(\nu_x^\top-\nu_y^\top)\bs\phi(y) -\pa_{\nu_x}G_{k_s}(x,y)\bs\phi(y) \right] ds_y,\quad x\in\Gamma.
\enn
\end{lemma}
\begin{remark}
The derived boundary integral equations (\ref{2DBIE}) and (\ref{3DBIE}) for two and three dimensional problems, respectively, take analogous forms. Therefore, for simplicity, the same notations $\mathbb{T}_\Gamma$ (and the following $\mathbb{T}_\Gamma^w, \mathbb{T}_\Pi, \mathbb{T}_\Pi^w$ in the next section) are used to avoid miscellaneous definitions of integral operators.
\end{remark}
\begin{proof}
On one hand, following the two-dimensional discussion and utilizing the boundary conditions in (\ref{3DHS1}), it holds that
\be
\label{3DscawaveP}
\phi_p^\mathrm{sca}(x) &=& -\mathcal{S}_p^\Gamma[\nu\cdot\nabla\times\phi_s^\mathrm{tot}](x)-\mathcal{D}_p^\Gamma[\phi_p^\mathrm{tot}](x),\quad x\in\Omega,
\en
and the following BIE results:
\be
\label{3DBIE1}
\frac{1}{2}\phi_p^\mathrm{tot}+K_p^\Gamma[\phi_p^\mathrm{tot}]+V_p^\Gamma[\nu\cdot\nabla\times\bs\phi_s^\mathrm{tot}] &=& \phi_p^\mathrm{inc} \quad\mbox{on}\quad\Gamma.
\en
Applying the Stokes theorem gives
\begin{align*}
V_p^\Gamma[\nu\cdot\nabla\times\bs\phi_s^\mathrm{tot}]&=\int_\Gamma G_{k_p}(x,y)\nu_y\cdot\left(\nabla_y\times\bs\phi_s^\mathrm{tot}(y)\right)ds_y\nonumber\\
&= -\int_\Gamma\nu_y\cdot \left(\nabla_yG_{k_p}(x,y)\times\bs\phi_s^\mathrm{tot}(y)\right) ds_y\nonumber\\
&= -\int_\Gamma \nabla_yG_{k_p}(x,y)\cdot \left(\bs\phi_s^\mathrm{tot}(y)\times\nu_y \right)ds_y\nonumber\\
&= -\int_\Gamma \nabla_{\Gamma,y}G_{k_p}(x,y)\cdot \left(\bs\phi_s^\mathrm{tot}(y)\times\nu_y \right)ds_y\nonumber.
\end{align*}
In view of the integration-by-part formula
\ben
\int_\Gamma \nabla_\Gamma\varphi_1 \cdot\bs \varphi_2 ds=- \int_\Gamma \varphi_1\mathrm{div}_\Gamma\bs \varphi_2 ds
\enn
for a scaler field $\varphi_1$ and a vector field $\bs \varphi_2$, we can obtain that
\ben
V_p^\Gamma[\nu\cdot\nabla\times\bs\phi_s^\mathrm{tot}]=
V_p^\Gamma[\mathrm{div}_\Gamma(\bs\phi_s^\mathrm{tot}\times\nu)].
\enn
Therefore, (\ref{3DBIE1}) is equivalent to
\be
\label{3DBIE1-equ}
\frac{1}{2}\phi_p^\mathrm{tot}+K_p^\Gamma[\phi_p^\mathrm{tot}]+ V_p^\Gamma[\mathrm{div}_\Gamma(\bs\phi_s^\mathrm{tot}\times\nu)] &=& \phi_p^\mathrm{inc} \quad\mbox{on}\quad\Gamma.
\en

On the other hand, from the Stratton-Chu integral representation formulae \cite[(9.8)]{P03} and the ideas put forth in \cite{DM98}, it can be shown that
\be
\label{3Dscawaverepre}
\bs\phi_s^\mathrm{sca}(x)=-\mathcal{M}^\Gamma_s[\nu\times \bs\phi_s^\mathrm{sca}](x) -\mathcal{N}^\Gamma_s[\nu\times \nabla\times\bs\phi_s^\mathrm{sca}](x),\quad x\in\Omega,
\en
where the layer potentials $\mathcal{M}_s^\Gamma, \mathcal{N}_s^\Gamma$ are defined by
\ben
&& \mathcal{M}^\Gamma_s[\bs \phi](x)= \nabla_x\times\int_\Gamma G_{k_s}(x,y)\bs\phi(y)ds_y,\quad x\notin\Gamma,\\
&& \mathcal{N}^\Gamma_s[\bs \phi](x)= \frac{1}{k_s^2}\nabla_x\times\nabla_x\times\int_\Gamma G_{k_s}(x,y)\bs\phi(y)ds_y,\quad x\notin\Gamma.
\enn
Additionally, the incident field $\bs\phi_s^\mathrm{inc}$ satisfies
\be
\label{3Dincwaverepre}
0=-\mathcal{M}^\Gamma_s[\nu\times \bs\phi_s^\mathrm{inc}](x) -\mathcal{N}^\Gamma_s[\nu\times \nabla\times\bs\phi_s^\mathrm{inc}](x),\quad x\in\Omega.
\en
Then combining (\ref{3Dscawaverepre})-(\ref{3Dincwaverepre}) and
applying the boundary conditions in (\ref{3DHS1}) yield
\be
\label{3DscawaveS}
\bs\phi_s^\mathrm{sca}(x)= \mathcal{M}^\Gamma_s[\bs\phi_s^\mathrm{tot}\times\nu](x) + \mathcal{N}^\Gamma_s[\bs\curl_\Gamma\bs\phi_p^\mathrm{tot}](x),\quad x\in\Omega,
\en
which further leads to
\be
\label{3DscawaveS1}
\bs\phi_s^\mathrm{sca}(z)\times\nu_x= -\nu_x\times \mathcal{M}^\Gamma_s[\bs\phi_s^\mathrm{tot}\times\nu](z) -\nu_x\times \mathcal{N}^\Gamma_s[\bs\curl_\Gamma\bs\phi_p^\mathrm{tot}](z),
\en
by operating $\times\nu_x$ on both sides of (\ref{3DscawaveS}) for $x\in\Gamma$ and $z=x-h\nu_x\in\Omega$. It follows that for vector tangential fields $\bs v$ (i.e., $\nu\cdot\bs v=0$),
\ben
\nu_x\times \mathcal{M}^\Gamma_s[\bs v](z) &=& \nu_x\times \nabla_z\times\int_\Gamma G_{k_s}(z,y)\bs v(y)ds_y \\
&=& \int_\Gamma \nu_x\times (\nabla_zG_{k_s}(z,y)\times \bs v(y))ds_y\\
&=& \int_\Gamma \left[\nabla_zG_{k_s}(z,y)\nu_x^\top\bs v(y) -(\nu_x\cdot \nabla_zG_{k_s}(z,y))\bs v(y) \right] ds_y\\
&=& \int_\Gamma \left[\nabla_zG_{k_s}(z,y)(\nu_x^\top-\nu_y^\top)\bs v(y) -(\nu_x\cdot \nabla_zG_{k_s}(z,y))\bs v(y) \right] ds_y,
\enn
and, in view of the identity $\nabla\times\nabla\times \bs w=\nabla\nabla\cdot\bs w-\Delta\bs w$,
\ben
\nu_x\times \mathcal{N}^\Gamma_s[\bs v](z) &=& \frac{1}{k_s^2}\nu_x\times\nabla_z\times\nabla_z\times\int_\Gamma G_{k_s}(z,y)\bs v(y)ds_y\\
&=& \frac{1}{k_s^2}\nu_x\times\nabla_z\int_\Gamma \nabla_zG_{k_s}(z,y)\cdot\bs v(y)ds_y +\nu_x\times \mathcal{S}_s^\Gamma[\bs v](z)  \\
&=& -\frac{1}{k_s^2}\nu_x\times\nabla_z\int_\Gamma \nabla_yG_{k_s}(z,y)\cdot\bs v(y)ds_y +\nu_x\times \mathcal{S}_s^\Gamma[\bs v](z) \\
&=& \frac{1}{k_s^2}\nu_x\times\nabla_z \int_\Gamma G_{k_s}(z,y) \mathrm{div}_\Gamma\bs v(y)ds_y +\nu_x\times \mathcal{S}_s^\Gamma[\bs v](z).
\enn
Noting that $\mathrm{div}_\Gamma \bs\curl_\Gamma=0$ \cite[Theorem 2.5.19]{N01}, we have
\ben
\nu_x\times \mathcal{N}^\Gamma_s[\bs\curl_\Gamma\bs\phi_p^\mathrm{tot}](z)= \nu_x\times \mathcal{S}_s^\Gamma[\bs\curl_\Gamma\bs\phi_p^\mathrm{tot}](z).
\enn
Taking the limit as $h\rightarrow 0^+$ ($z\rightarrow x$) in (\ref{3DscawaveS1}) and using the well-known jump relations associated with the acoustic layer potentials~\cite{CK13}, we obtain
\begin{align}
\label{3DBIE2}
\frac{1}{2}\bs\phi_s^\mathrm{tot}\times\nu+
M_s^\Gamma[\bs\phi_s^\mathrm{tot}\times\nu] + \nu\times V_s^\Gamma[\bs\curl_\Gamma\phi_p^\mathrm{tot}]= \bs\phi_s^\mathrm{inc}\times\nu \quad\mbox{on}\quad\Gamma.
\end{align}
Incorporating (\ref{3DBIE1}) and (\ref{3DBIE2}) leads to (\ref{3DBIE}) and the proof is completed.
\end{proof}

\section{HD-WGF method}
\label{sec:4}

The derived boundary integral equations (\ref{2DBIE}) and (\ref{3DBIE}) are defined on the unbounded surface $\Gamma=\Pi\cup \pa D$, and thus, their numerical discretization requires appropriate truncation of $\Gamma$. Utilizing the WGF strategy proposed in~\cite{P16} for acoustic layered-media scattering, this section proposes new HD-WGF methods for solving the elastic scattering problems on a half-space. The methods rely on an effective domain truncation strategy using a smooth ``slow-rise'' windowing function $w_A$ which vanishes outside an interval of length $2A$, equals one in a region around the origin which grows linearly with $A$, and has a slow rise: all of its derivatives {tend} to zero uniformly as $A\rightarrow +\infty$. For example, as used in this work, we can define $w_A$ as
\ben
w_A(t)=\eta(t/A;c,1),
\enn
where $0<c<1$ and
\ben
\eta(t;t_0,t_1)=\begin{cases} 1,
  & |t|\le t_0, \cr e^{\frac{2e^{-1/u}}{u-1}}, & t_0<|t|<t_1,
  u=\frac{|t|-t_0}{t_1-t_0}, \cr 0, & |t|\ge t_1.
\end{cases}
\enn
An alternative choice of the windowing function can be found in \cite{LGO18}. To introduce our WGF methods, we define the windowing operator $\sqcap_A$ as $\sqcap_A[\phi]:=\widetilde{w}_A\phi$ where
\ben
\widetilde{w}_A(x)=\begin{cases}
w_A(x_1) , & \mbox{in}\quad 2D,\cr
w_A(x_1)w_A(x_2), & \mbox{in}\quad 3D,
\end{cases}
\enn
and denote by $\Gamma_A=\Gamma\cap\{\widetilde{w}_A\ne 0\}$ the windowed {finite} part of the boundary $\Gamma$.

\subsection{Two-dimensional case}
\label{sec:4.1}

Utilizing the windowing operator $\sqcap_A$, it is known that the boundary integral equation (\ref{2DBIE}) is equivalent to
\be
\label{2DWIEM}
\left(\frac{\mathbb{I}}{2}+\mathbb{T}^w_\Gamma\right)\begin{bmatrix}
\phi_p^\mathrm{tot}\\
\phi_s^\mathrm{tot}
\end{bmatrix}=\begin{bmatrix}
\phi_p^\mathrm{inc}\\
\phi_s^\mathrm{inc}
\end{bmatrix}+\left(\mathbb{T}^w_\Gamma-\mathbb{T}_\Gamma\right)\begin{bmatrix}
\phi_p^\mathrm{tot}\\
\phi_s^\mathrm{tot}
\end{bmatrix}\quad\mbox{on}\quad\Gamma,
\en
where
\ben
\mathbb{T}^w_\Gamma=\begin{bmatrix}
K_p^\Gamma\sqcap_A & V_p^\Gamma\pa_\tau \sqcap_A \\
-V_s^\Gamma\pa_\tau \sqcap_A & K_s^\Gamma\sqcap_A
\end{bmatrix}.
\enn
Then eliminating the second term on the right-hand side of (\ref{2DWIEM}) gives a simple windowed version of (\ref{2DBIE}), i.e.,
\be
\label{2DWIE}
\left(\frac{\mathbb{I}}{2}+\mathbb{T}^w_\Gamma\right)\begin{bmatrix}
\widehat\phi_p^\mathrm{tot}\\
\widehat\phi_s^\mathrm{tot}
\end{bmatrix}=\begin{bmatrix}
\phi_p^\mathrm{inc}\\
\phi_s^\mathrm{inc}
\end{bmatrix} \quad\mbox{on}\quad\Gamma_A,
\en
and the following windowed solution representations can be utilized
\be
\label{2DscawavePDW}
\phi_p^\mathrm{sca}(x) &=& -\mathcal{S}_p^\Gamma[\partial_\tau\sqcap_A\widehat\phi_s^\mathrm{tot}](x)-\mathcal{D}_p^\Gamma[\sqcap_A\widehat\phi_p^\mathrm{tot}](x),\quad x\in\Omega,\\
\label{2DscawaveSDW}
\phi_s^\mathrm{sca}(x) &=& \mathcal{S}_s^\Gamma[\partial_\tau\sqcap_A\widehat\phi_p^\mathrm{tot}](x)-\mathcal{D}_s^\Gamma[\sqcap_A\widehat\phi_s^\mathrm{tot}](x),\quad x\in\Omega.
\en
However, analogous to the illustration in \cite{BLPT16,BY21}, this direct windowing approach will not be uniformly accurate with respect to the angle of incidence. It can be expected that the approximation resulting from (\ref{2DscawavePDW}) and (\ref{2DscawaveSDW}) require, for a given accuracy, increasingly large truncated domains, i.e., large $A$, as grazing incidence is approached ($\theta_\xi^\mathrm{inc}\rightarrow \pi/2, \xi=p,s$). As noted in \cite{BLPT16}, this difficulty can be explained by consideration of certain arguments concerning bouncing geometrical optics rays and the method of stationary phase.

Inspired by the ``corrected'' formulations discussed in \cite{BLPT16} for acoustic layer problems, we now propose a new ``corrected'' windowed BIE and the associated solution representations for the two-dimensional elastic scattering problems on a half-space. On $\Gamma=\Pi\cup\partial D$, we introduce the operators $\mathbb{T}_\Pi$ and $\mathbb{T}^w_\Pi$ defined as follows:
\ben
\mathbb{T}_\Pi=\begin{bmatrix}
K_p^\Pi & V_p^\Pi\pa_\tau  \\
-V_s^\Pi\pa_\tau & K_s^\Pi
\end{bmatrix},\quad \mathbb{T}^w_\Pi=\begin{bmatrix}
K_p^\Pi\sqcap_A & V_p^\Pi\pa_\tau \sqcap_A \\
-V_s^\Pi\pa_\tau \sqcap_A & K_s^\Pi\sqcap_A
\end{bmatrix}\quad\mbox{on}\quad \Pi,
\enn
and
\ben
\mathbb{T}_\Pi=\begin{bmatrix}
\mathcal{D}_p^\Pi & \mathcal{S}_p^\Pi\pa_\tau  \\
-\mathcal{S}_s^\Pi\pa_\tau & \mathcal{D}_s^\Pi
\end{bmatrix},\quad \mathbb{T}^w_\Pi=\begin{bmatrix}
\mathcal{D}_p^\Pi\sqcap_A & \mathcal{S}_p^\Pi\pa_\tau \sqcap_A \\
-\mathcal{S}_s^\Pi\pa_\tau \sqcap_A & \mathcal{D}_s^\Pi\sqcap_A
\end{bmatrix}\quad\mbox{on}\quad \partial D.
\enn
Then (\ref{2DWIEM}) is equivalent to
\be
\label{2DWIEM1}
\left(\frac{\mathbb{I}}{2}+\mathbb{T}^w\right)\begin{bmatrix}
\phi_p^\mathrm{tot}\\
\phi_s^\mathrm{tot}
\end{bmatrix}=\begin{bmatrix}
\phi_p^\mathrm{inc}\\
\phi_s^\mathrm{inc}
\end{bmatrix}+\left(\mathbb{T}^w_\Pi-\mathbb{T}_\Pi\right)\begin{bmatrix}
\phi_p^\mathrm{tot}\\
\phi_s^\mathrm{tot}
\end{bmatrix}\quad\mbox{on}\quad\Gamma.
\en
Note that the second term on the right-hand side only depends on the values of $\phi_p^\mathrm{tot}, \phi_s^\mathrm{tot}$ on $\{x\in\Pi: \widetilde{w}_A(x)\ne 1\}$ which can be approximated by the solutions, denoted by $\phi_{p,\Pi}^\mathrm{tot}, \phi_{s,\Pi}^\mathrm{tot}$ (see Appendix~\ref{sec:A}), of the problem of elastic scattering by (only) the flat infinite surface $\Pi$ for sufficiently large $A$ from the view point of the physical concept underlying the scattering of plane incidence. Thus, replacing $\phi_p^\mathrm{tot}, \phi_s^\mathrm{tot}$ on the right hand side of (\ref{2DWIEM1}) by $\phi_{p,\Pi}^\mathrm{tot}, \phi_{s,\Pi}^\mathrm{tot}$ we arrive at the following ``corrected'' windowed BIE for the two-dimensional elastic scattering problem on a half-space:
\be
\label{2DWIEM-lim}
\left(\frac{\mathbb{I}}{2}+\mathbb{T}^w\right)\begin{bmatrix}
\widetilde\phi_p^\mathrm{tot}\\
\widetilde\phi_s^\mathrm{tot}
\end{bmatrix}=\begin{bmatrix}
\phi_p^\mathrm{inc}\\
\phi_s^\mathrm{inc}
\end{bmatrix}+\left(\mathbb{T}^w_\Pi-\mathbb{T}_\Pi\right)\begin{bmatrix}
\phi_{p,\Pi}^\mathrm{tot}\\
\phi_{s,\Pi}^\mathrm{tot}
\end{bmatrix}\quad\mbox{on}\quad\Gamma_A.
\en

\begin{lemma}
\label{lemma4.1}
It holds that
\be
\label{closedform2D1}
\mathbb{T}_{\Pi}\begin{bmatrix}
\phi_{p,\Pi}^\mathrm{tot}\\
\phi_{s,\Pi}^\mathrm{tot}
\end{bmatrix}=-\begin{bmatrix}
\phi_{p,\Pi}^\mathrm{sca}\\
\phi_{s,\Pi}^\mathrm{sca}
\end{bmatrix}\quad\mbox{on}\quad \pa D,
\en
and
\be
\label{closedform2D2}
\mathbb{T}_{\Pi}\begin{bmatrix}
\phi_{p,\Pi}^\mathrm{tot}\\
\phi_{s,\Pi}^\mathrm{tot}
\end{bmatrix}=\begin{bmatrix}
\phi_p^\mathrm{inc}-\frac{1}{2}\phi_{p,\Pi}^\mathrm{tot}\\
\phi_s^\mathrm{inc}-\frac{1}{2}\phi_{s,\Pi}^\mathrm{tot}
\end{bmatrix}\quad\mbox{on}\quad\Pi.
\en
\end{lemma}
\begin{proof}
The proof of this lemma can be carried out by means of the solution representations associated with the problem of elastic scattering by the flat infinite surface. It is known that, for $\xi=p,s$,
\ben
0=&\mathcal{S}_\xi^\Pi[\partial_\nu\phi_\xi^\mathrm{inc}](x)-\mathcal{D}_\xi^\Pi[\phi_\xi^\mathrm{inc}](x)\quad\mbox{in}\quad\mathbb{R}^2_+,\\
\phi_{\xi,\Pi}^\mathrm{sca}=&\mathcal{S}_\xi^\Pi[\partial_\nu\phi_{\xi,\Pi}^\mathrm{sca}](x)-\mathcal{D}_\xi^\Pi[\phi_{\xi,\Pi}^\mathrm{sca}](x)\quad\mbox{in}\quad\mathbb{R}^2_+.
\enn
Therefore, the boundary conditions
\ben
\pa_\nu\phi_{p,\Pi}^\mathrm{tot}+\pa_\tau\phi_{s,\Pi}^\mathrm{tot}=0,\quad \pa_\tau\phi_{p,\Pi}^\mathrm{tot}-\pa_\nu\phi_{s,\Pi}^\mathrm{tot}=0\quad \mbox{on}\quad \Pi
\enn
implies that
\be
\label{lemma4.1-1}
\phi_{p,\Pi}^\mathrm{sca}=&-\mathcal{S}_p^\Pi[\partial_\tau\phi_{s,\Pi}^\mathrm{tot}](x)-\mathcal{D}_p^\Pi[\phi_{p,\Pi}^\mathrm{tot}](x)\quad\mbox{in}\quad\mathbb{R}^2_+.
\en
and
\be
\label{lemma4.1-2}
\phi_{s,\Pi}^\mathrm{sca}=&\mathcal{S}_s^\Pi[\partial_\tau\phi_{p,\Pi}^\mathrm{tot}](x)-\mathcal{D}_s^\Pi[\phi_{s,\Pi}^\mathrm{tot}](x)\quad\mbox{in}\quad\mathbb{R}^2_+,
\en
which give (\ref{closedform2D1}). Taking the limit $x\rightarrow\Pi$ in (\ref{lemma4.1-1}) and (\ref{lemma4.1-2}) will lead to (\ref{closedform2D2}) which completes the proof.
\end{proof}

\begin{remark}
\label{localdefect2D}
The more general consideration of the case that the unbounded rough surface $\Gamma$ consists both a local perturbation of the flat surface $\Pi$ and/or the boundary of bounded impenetrable obstacles in $U_0$ will not bring any significant challenge. In addition to Lemma~\ref{lemma4.1}, it is only necessary to treat the term $\mathbb{T}_{\Pi}[\phi_{p,\Pi}^\mathrm{tot};
\phi_{s,\Pi}^\mathrm{tot}]$ on $(\Gamma\backslash \Pi)\cap\R^d_-$ by means of the solution representations
\begin{align*}
-\phi_\xi^\mathrm{inc}=&\mathcal{S}_\xi^\Pi[\partial_\nu\phi_\xi^\mathrm{inc}](x)-\mathcal{D}_\xi^\Pi[\phi_\xi^\mathrm{inc}](x)\quad\mbox{in}\quad\mathbb{R}^d_-,\\
0=&\mathcal{S}_\xi^\Pi[\partial_\nu\phi_{\xi,\Pi}^\mathrm{sca}](x)-\mathcal{D}_\xi^\Pi[\phi_{\xi,\Pi}^\mathrm{sca}](x)\quad\mbox{in}\quad\mathbb{R}^d_-.
\end{align*}
Then it can be derived that
\ben
\mathbb{T}_{\Pi}\begin{bmatrix}
\phi_{p,\Pi}^\mathrm{tot}\\
\phi_{s,\Pi}^\mathrm{tot}
\end{bmatrix}=\begin{bmatrix}
\phi_p^\mathrm{inc}\\
\phi_s^\mathrm{inc}
\end{bmatrix}\quad\mbox{on}\quad (\Gamma\backslash \Pi)\cap\R^d_-.
\enn
\end{remark}

Taking (\ref{closedform2D1})-(\ref{closedform2D2}) into (\ref{2DWIEM-lim}), we finally arrive at the approximation of the solutions $\phi_\xi^\mathrm{tot},\xi=p,s$ on $\Gamma_A$ using $\widetilde\phi_\xi^\mathrm{tot},\xi=p,s$ which {satisfy} the boundary integral equation
\be
\label{2DWIEM-corrected}
\left(\frac{\mathbb{I}}{2}+\mathbb{T}^w\right)\begin{bmatrix}
\widetilde\phi_p^\mathrm{tot}\\
\widetilde\phi_s^\mathrm{tot}
\end{bmatrix}=\mathbb{T}_{\Pi}^w\begin{bmatrix}
\phi_{p,\Pi}^\mathrm{tot}\\
\phi_{s,\Pi}^\mathrm{tot}
\end{bmatrix}+\begin{cases}
[\phi_{p,\Pi}^\mathrm{tot}; \phi_{s,\Pi}^\mathrm{tot}] & \quad\mbox{on}\quad\partial D,\\
[\frac{1}{2}\phi_{p,\Pi}^\mathrm{tot}; \frac{1}{2}\phi_{s,\Pi}^\mathrm{tot}]  &\quad\mbox{on}\quad\Pi.\\
\end{cases}
\en
Recall the solution representations given in (\ref{2DscawaveP})-(\ref{2DscawaveS}):
\ben
\phi_p^\mathrm{sca}(x) &=& -\mathcal{S}_p^\Gamma[\partial_\tau\phi_s^\mathrm{tot}](x)-\mathcal{D}_p^\Gamma[\phi_p^\mathrm{tot}](x),\quad x\in\Omega,\\
\phi_s^\mathrm{sca}(x) &=& \mathcal{S}_s^\Gamma[\partial_\tau\phi_p^\mathrm{tot}](x)-\mathcal{D}_s^\Gamma[\phi_s^\mathrm{tot}](x),\quad x\in\Omega.
\enn
Noting that $\phi_\xi^\mathrm{tot}=\sqcap_A \phi_\xi^\mathrm{tot}+ (\mathbb{I}-\sqcap_A)\phi_\xi^\mathrm{tot}$, then the approximation strategies $\phi_\xi^\mathrm{tot}\approx \widetilde\phi_\xi^\mathrm{tot}$ on $\Gamma_A$ and $\phi_\xi^\mathrm{tot}\approx \phi_{\xi,\Pi}^\mathrm{tot}$ on $\{x\in\Pi: \widetilde{w}_A(x)\ne 1\}$, together with (\ref{lemma4.1-1})-(\ref{lemma4.1-2}), lead to the approximation of the solutions $\phi_p^\mathrm{sca},\xi=p,s$ in $\Omega$ as follows:
\be
\label{2DscawaveP-corrected}
\phi_p^\mathrm{sca}(x) &\approx& -\mathcal{S}_p^\Gamma[\partial_\tau \sqcap_A \widetilde\phi_s^\mathrm{tot}](x)-\mathcal{D}_p^\Gamma[\sqcap_A \widetilde\phi_p^\mathrm{tot}](x)\nonumber\\
&\quad& +\mathcal{S}_p^\Pi[\partial_\tau \sqcap_A \phi_{s,\Pi}^\mathrm{tot} ](x) +\mathcal{D}_p^\Pi[\sqcap_A \phi_{p,\Pi}^\mathrm{tot}](x)\nonumber\\
&\quad& +\phi_{p,\Pi}^\mathrm{sca}(x),\quad x\in\Omega,
\en
and
\be
\label{2DscawaveS-corrected}
\phi_s^\mathrm{sca}(x) &\approx& \mathcal{S}_s^\Gamma[\partial_\tau \sqcap_A \widetilde\phi_p^\mathrm{tot}](x)-\mathcal{D}_s^\Gamma[\sqcap_A \widetilde\phi_s^\mathrm{tot}](x)\nonumber\\
&\quad& -\mathcal{S}_s^\Pi[\partial_\tau \sqcap_A \phi_{p,\Pi}^\mathrm{tot} ](x) +\mathcal{D}_s^\Pi[\sqcap_A \phi_{s,\Pi}^\mathrm{tot}](x)\nonumber\\
&\quad& +\phi_{s,\Pi}^\mathrm{sca}(x),\quad x\in\Omega.
\en

\subsection{Three-dimensional case}
\label{sec:4.2}

The windowed boundary integral equation for three-dimensional problem can be generated analogous to the approach for two-dimensional problem. Now we denote
\ben
\mathbb{T}_\Gamma^w=\begin{bmatrix}
K_p^\Gamma\sqcap_A & V_p^\Gamma\mathrm{div}_\Gamma\sqcap_A \\
N_s^\Gamma\bs\curl_\Gamma\sqcap_A & M_s^\Gamma\sqcap_A
\end{bmatrix},
\enn
and introduce the operators $\mathbb{T}_\Pi$ and $\mathbb{T}_\Pi^w$ defined on $\Pi$ as follows:
\ben
\mathbb{T}_\Pi=\begin{bmatrix}
K_p^\Pi & V_p^\Pi\mathrm{div}_\Pi \\
N_s^\Pi\bs\curl_\Pi & M_s^\Pi
\end{bmatrix},\quad \mathbb{T}_\Pi^w =\begin{bmatrix}
K_p^\Pi\sqcap_A & V_p^\Pi\mathrm{div}_\Pi\sqcap_A \\
N_s^\Pi\bs\curl_\Pi\sqcap_A & M_s^\Pi\sqcap_A
\end{bmatrix}\quad \mbox{on}\quad\Pi,
\enn
and
\ben
\mathbb{T}_\Pi=\begin{bmatrix}
\mathcal{D}_p^\Pi & \mathcal{S}_p^\Pi\mathrm{div}_\Pi \\
\nu\times\mathcal{N}_s^\Pi\bs\curl_\Pi & \nu\times\mathcal{M}_s^\Pi
\end{bmatrix},\quad \mathbb{T}_\Pi^w =\begin{bmatrix}
\mathcal{D}_p^\Pi\sqcap_A & \mathcal{S}_p^\Pi\mathrm{div}_\Pi\sqcap_A \\
\nu\times\mathcal{N}_s^\Pi\bs\curl_\Pi\sqcap_A & \nu\times\mathcal{M}_s^\Pi\sqcap_A
\end{bmatrix}\quad \mbox{on}\quad\partial D.
\enn
Then the boundary integral equation (\ref{3DBIE}) is equivalent to
\begin{align}
\label{3DWIEM}
\left(\frac{\mathbb{I}}{2}+\mathbb{T}_\Gamma^w\right)\begin{bmatrix}
\phi_p^\mathrm{tot}\\
\bs\phi_s^\mathrm{tot}\times\nu
\end{bmatrix}&=\begin{bmatrix}
\phi_p^\mathrm{inc}\\
\bs\phi_s^\mathrm{inc}\times\nu
\end{bmatrix}+ (\mathbb{T}_\Gamma^w-\mathbb{T}_\Gamma)\begin{bmatrix}
\phi_p^\mathrm{tot}\\
\bs\phi_s^\mathrm{tot}\times\nu
\end{bmatrix}\nonumber\\
&=\begin{bmatrix}
\phi_p^\mathrm{inc}\\
\bs\phi_s^\mathrm{inc}\times\nu
\end{bmatrix}+ (\mathbb{T}_\Pi^w-\mathbb{T}_\Pi)\begin{bmatrix}
\phi_p^\mathrm{tot}\\
\bs\phi_s^\mathrm{tot}\times\nu
\end{bmatrix}
 \quad\mbox{on}\quad\Gamma.
\end{align}
Let $\phi_{p,\Pi}^\mathrm{tot}, \bs\phi_{s,\Pi}^\mathrm{tot}$ (see Appendix A) be the solutions of the problem of elastic scattering by (only) the flat infinite surface $\Pi$ in three dimensions. Following the correction idea in (\ref{2DWIEM-lim}) and replacing $\phi_p^\mathrm{tot}, \bs\phi_s^\mathrm{tot}$ on the right hand side of (\ref{3DWIEM}) by $\phi_{p,\Pi}^\mathrm{tot}, \bs\phi_{s,\Pi}^\mathrm{tot}$ we arrive at the following ``corrected'' windowed BIE for the three-dimensional elastic scattering problem on a half-space:
\begin{align}
\label{3DWIEM-lim}
\left(\frac{\mathbb{I}}{2}+\mathbb{T}_\Gamma^w\right)\begin{bmatrix}
\phi_p^\mathrm{tot}\\
\bs\phi_s^\mathrm{tot}\times\nu
\end{bmatrix} =\begin{bmatrix}
\phi_p^\mathrm{inc}\\
\bs\phi_s^\mathrm{inc}\times\nu
\end{bmatrix}+ (\mathbb{T}_\Pi^w-\mathbb{T}_\Pi)\begin{bmatrix}
\phi_{p,\Pi}^\mathrm{tot}\\
\bs\phi_{s,\Pi}^\mathrm{tot}\times\nu
\end{bmatrix}
 \quad\mbox{on}\quad\Gamma.
\end{align}

The following closed-form can be shown {analogous} to Lemma~\ref{lemma4.1} together with the solution representations of Maxwell equation and the proof is omitted.
\begin{lemma}
\label{lemma5.1}
It holds that
\be
\label{closedform3D1}
\mathbb{T}_{\Pi}\begin{bmatrix}
\phi_{p,\Pi}^\mathrm{tot}\\
\bs\phi_{s,\Pi}^\mathrm{tot}\times\nu
\end{bmatrix}=-\begin{bmatrix}
\phi_{p,\Pi}^\mathrm{sca}\\
\bs\phi_{s,\Pi}^\mathrm{sca}\times\nu
\end{bmatrix}\quad\mbox{on}\quad \pa D,
\en
and
\be
\label{closedform3D2}
\mathbb{T}_{\Pi}\begin{bmatrix}
\phi_{p,\Pi}^\mathrm{tot}\\
\bs\phi_{s,\Pi}^\mathrm{tot}\times\nu
\end{bmatrix}=\begin{bmatrix}
\phi_p^\mathrm{inc}-\frac{1}{2}\phi_{p,\Pi}^\mathrm{tot}\\
\bs\phi_s^\mathrm{inc}\times\nu-\frac{1}{2}\bs\phi_{s,\Pi}^\mathrm{tot}\times\nu
\end{bmatrix}\quad\mbox{on}\quad\Pi.
\en
\end{lemma}

\section{Numerical experiments}
\label{sec:5}

In this section, we present several numerical experiments to demonstrate the efficiency and accuracy of the proposed HD-WGF method. Analogous to~\cite{BY21}, solutions for the integral equations were produced by means of the Chebyshev-based rectangular-polar discretization methodology~\cite{BG20,BY20} for the numerical evaluation of integral operators and the fully complex version of the iterative solver GMRES. The parameters in the numerical evaluation of the integral operators are selected such that the errors arising from the numerical integration are negligible in comparison with the smooth-windowing truncation errors. All of the numerical tests were obtained by means of Fortran numerical implementations, parallelized using OpenMP. In all cases, unless otherwise stated, the values $\lambda=2$, $\mu=1$,
$\rho=1$, $c=0.7$ were used and the relative errors reported
were calculated in accordance with the expression \be\label{rel_error}
\epsilon_\infty=\frac{\max_{x\in
    S}|u^{\mathrm{num}}(x)-u^{\mathrm{ref}}(x)|}{\max_{x\in
    S}|u^{\mathrm{ref}}(x)|}, \en where, denoting $u=(\phi_p^{\mathrm{sca}},\phi_s^{\mathrm{sca}})$ in 2D and $u=(\phi_p^{\mathrm{sca}},\bs\phi_s^{\mathrm{sca}})$ in 3D, $u^{\mathrm{num}}$ is the corresponding numerical solution and $u^{\mathrm{ref}}$ is produced by
means of numerical solution with a sufficiently fine discretization
and a sufficiently large value of $A$, and where $S$ is a suitably
selected line segment (2D) or square plane (3D) above the defect.

\begin{figure}[htbp]
\centering
\begin{tabular}{ccc}
\includegraphics[scale=0.15]{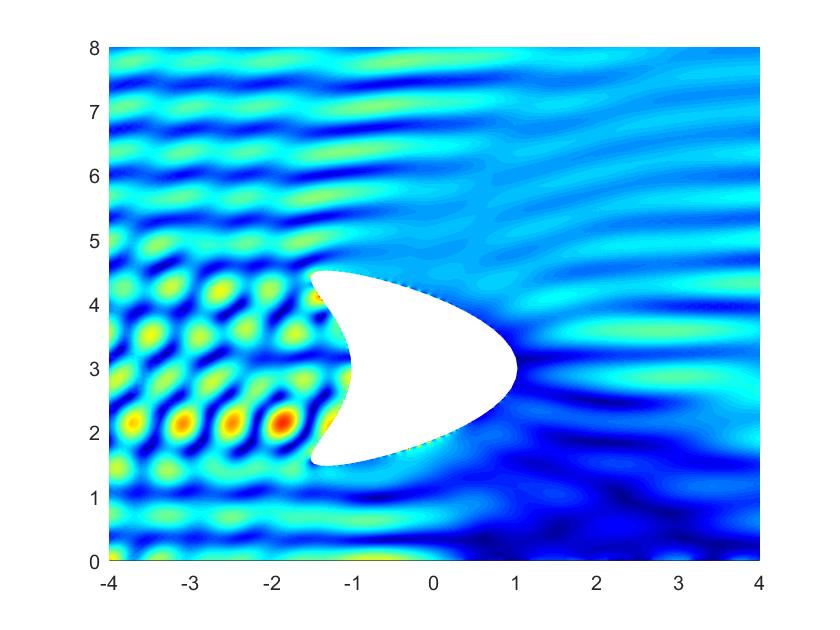} &
\includegraphics[scale=0.15]{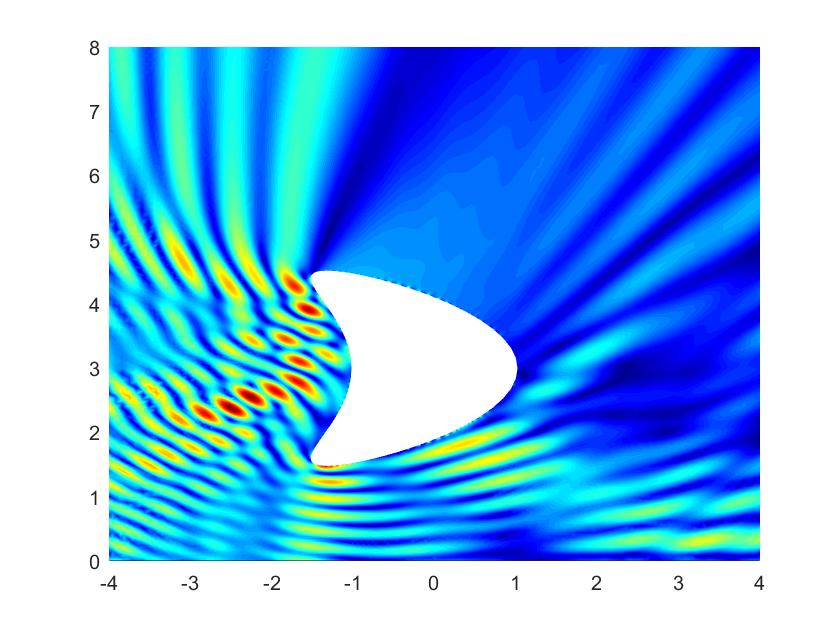} \\
(a) HD-WGF: $\phi_p^{\mathrm{tot}}$ & (b) HD-WGF: $\phi_s^{\mathrm{tot}}$ \\
\includegraphics[scale=0.15]{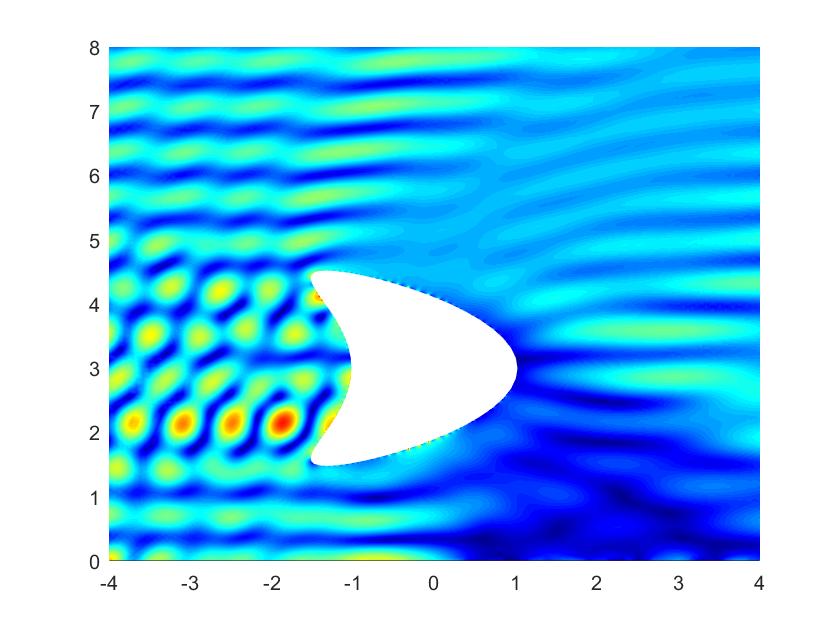} &
\includegraphics[scale=0.15]{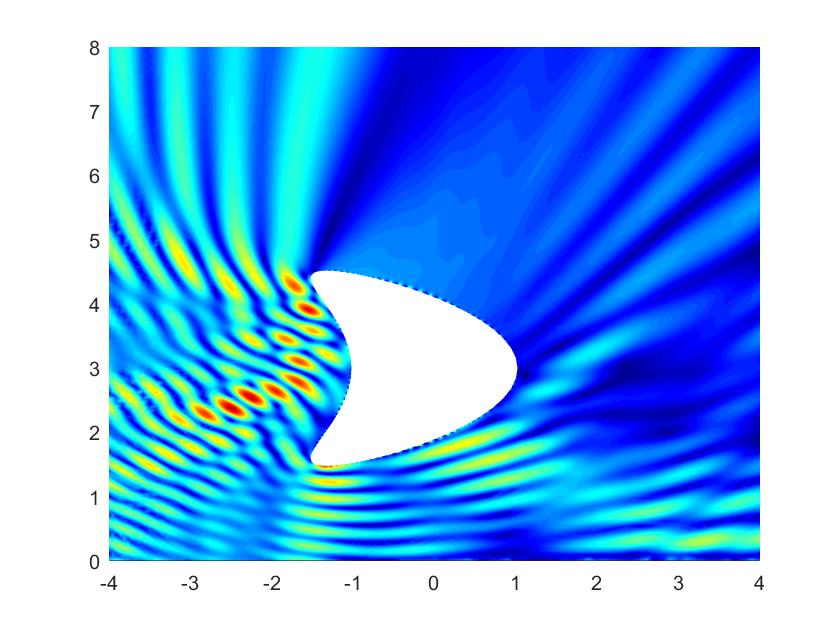} \\
(c)WGF: $\phi_p^{\mathrm{tot}}$ & (d)WGF: $\phi_s^{\mathrm{tot}}$ \\
\end{tabular}
\caption{Absolute values of the total fields $\phi_p^{\mathrm{tot}}$ and $\phi_s^{\mathrm{tot}}$ resulting from the WGF method proposed in~\cite{BY21} and the HD-WGF method for the problem of scattering by a kite-shaped obstacle where $\omega=4\pi$, $A=16\lambda_s$, $\theta_p^{\mathrm{inc}}=\frac{\pi}{4}$, $c_p=1$ and $c_s=0$.}
\label{Figure2D.1}
\end{figure}

\begin{table}[htbp]
\centering
\begin{tabular}{|c|c|c|c|c|c|c|}
\hline
\multicolumn{7}{|c|}{Disc-shaped} \\
\hline
\multirow{2}{*}{$A/\lambda_s$} & \multicolumn{3}{|c|}{$\theta_s^\mathrm{inc}=0$}  & \multicolumn{3}{|c|}{$\theta_s^\mathrm{inc}=-\frac{63\pi}{128}$}\\
\cline{2-7}
& $\theta_p^\mathrm{inc}=0$ & $\theta_p^\mathrm{inc}=\frac{\pi}{4}$ & $\theta_p^\mathrm{inc}=\frac{63\pi}{128}$
& $\theta_p^\mathrm{inc}=0$ & $\theta_p^\mathrm{inc}=\frac{\pi}{4}$ & $\theta_p^\mathrm{inc}=\frac{63\pi}{128}$  \\
\hline
2  & 5.92E-2 & 5.98E-2 & 2.84E-2 & 1.85E-2 & 2.82E-2 & 8.50E-3 \\
4  & 7.57E-3 & 5.79E-3 & 2.26E-3 & 4.41E-3 & 2.63E-3 & 1.18E-3 \\
8  & 4.29E-4 & 2.15E-4 & 1.17E-4 & 1.11E-4 & 9.45E-5 & 2.75E-5 \\
16 & 6.32E-6 & 7.57E-6 & 2.76E-6 & 1.49E-5 & 8.94E-6 & 7.73E-6 \\
\hline
\multicolumn{7}{|c|}{Kite-shaped} \\
\hline
\multirow{2}{*}{$A/\lambda_s$} & \multicolumn{3}{|c|}{$\theta_s^\mathrm{inc}=0$}  & \multicolumn{3}{|c|}{$\theta_s^\mathrm{inc}=-\frac{63\pi}{128}$}\\
\cline{2-7}
& $\theta_p^\mathrm{inc}=0$ & $\theta_p^\mathrm{inc}=\frac{\pi}{4}$ & $\theta_p^\mathrm{inc}=\frac{63\pi}{128}$
& $\theta_p^\mathrm{inc}=0$ & $\theta_p^\mathrm{inc}=\frac{\pi}{4}$ & $\theta_p^\mathrm{inc}=\frac{63\pi}{128}$  \\
\hline
2  & 1.59E-2 & 6.21E-2 & 1.87E-2 & 1.12E-2 & 3.33E-2 & 4.55E-3 \\
4  & 4.67E-3 & 1.08E-2 & 5.66E-3 & 1.84E-3 & 3.88E-3 & 1.08E-3 \\
8  & 3.17E-5 & 3.22E-4 & 8.70E-5 & 6.41E-5 & 1.67E-4 & 4.45E-5 \\
16 & 9.88E-6 & 2.64E-5 & 1.08E-5 & 4.70E-6 & 1.20E-5 & 2.86E-6 \\
\hline
\end{tabular}
\caption{Relative errors $\epsilon_\infty$ in the total field  resulting from the HD-WGF method for the problems of scattering by a disc-shaped and a kite-shaped obstacle within a half space where $c_p=1$ and $c_s=\sqrt{2}$.}
\label{Table2D.1}
\end{table}

{\bf Two-dimensional examples.} First of all, we consider the problem of scattering by a kite-shaped obstacle within a half space and choose $\omega=4\pi$, $A=16\lambda_s$, $\theta_p^{\mathrm{inc}}=\frac{\pi}{4}$, $c_p=1$ and $c_s=0$. Figures~\ref{Figure2D.1}(a,b) present the total fields $\phi_p^{\mathrm{tot}}$ and $\phi_s^{\mathrm{tot}}$ resulting from the HD-WGF method. As a comparison, the corresponding total fields $\phi_p^{\mathrm{tot}}$ and $\phi_s^{\mathrm{tot}}$ resulting from the WGF method proposed in~\cite{BY21} are presented in Figures~\ref{Figure2D.1}(c,d) which show the efficiency of the new HD-WGF solver. Choosing $c_p=1$ and $c_s=\sqrt{2}$, Table~\ref{Table2D.1} displays the relative errors in the total field that result from use of the proposed HD-WGF method for the case of a disc-shaped or a kite-shaped obstacle within a half space, clearly demonstrating the uniform fast convergence of the proposed approach over wide angular variations, going from normal incidence to grazing. The near fields for
the problem of scattering by the a kite-shaped obstacle for different incident angles are presented
in Figure~\ref{Figure2D.2}, respectively.

\begin{figure}[htbp]
\centering
\begin{tabular}{ccc}
\includegraphics[scale=0.15]{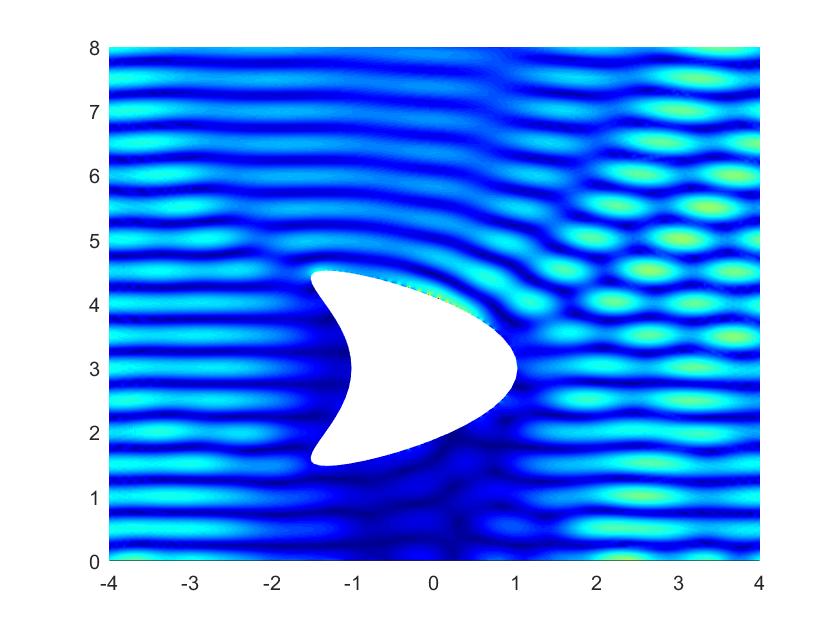} &
\includegraphics[scale=0.15]{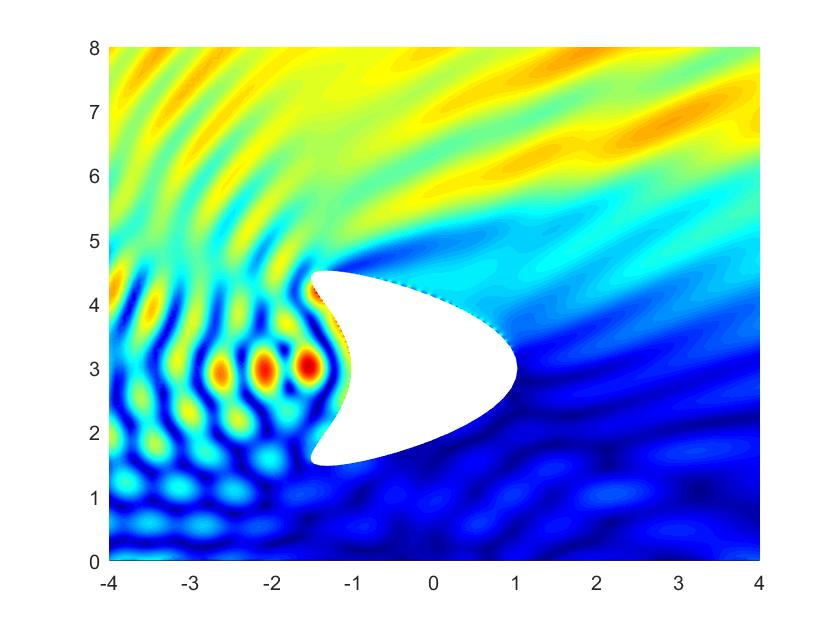} &
\includegraphics[scale=0.15]{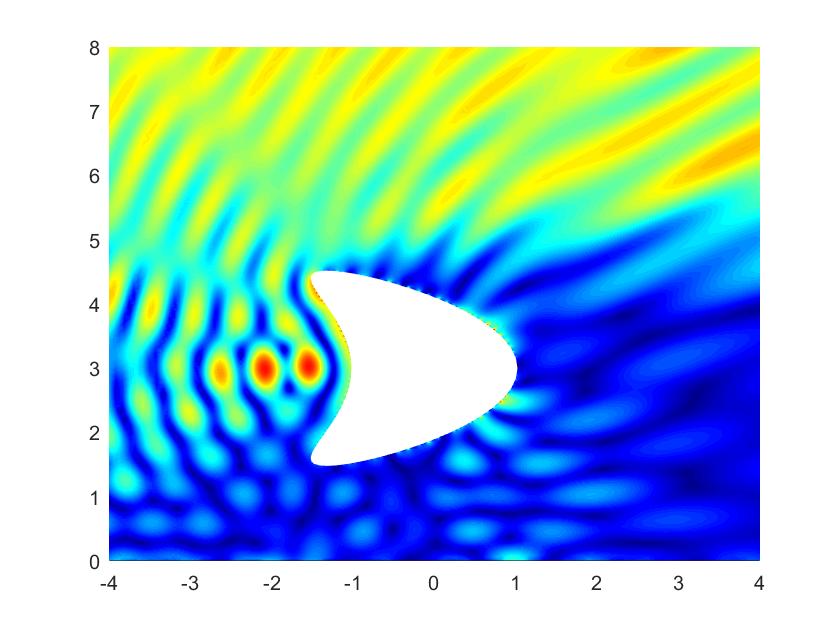} \\
(a) $\phi_p^{\mathrm{tot}}$ & (b) $\phi_p^{\mathrm{tot}}$ & (c) $\phi_p^{\mathrm{tot}}$ \\
\includegraphics[scale=0.15]{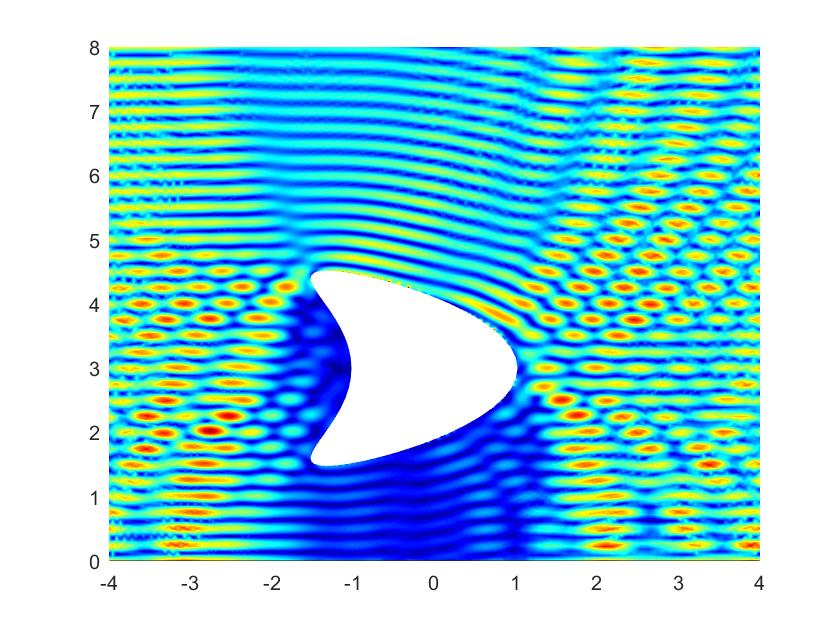} &
\includegraphics[scale=0.15]{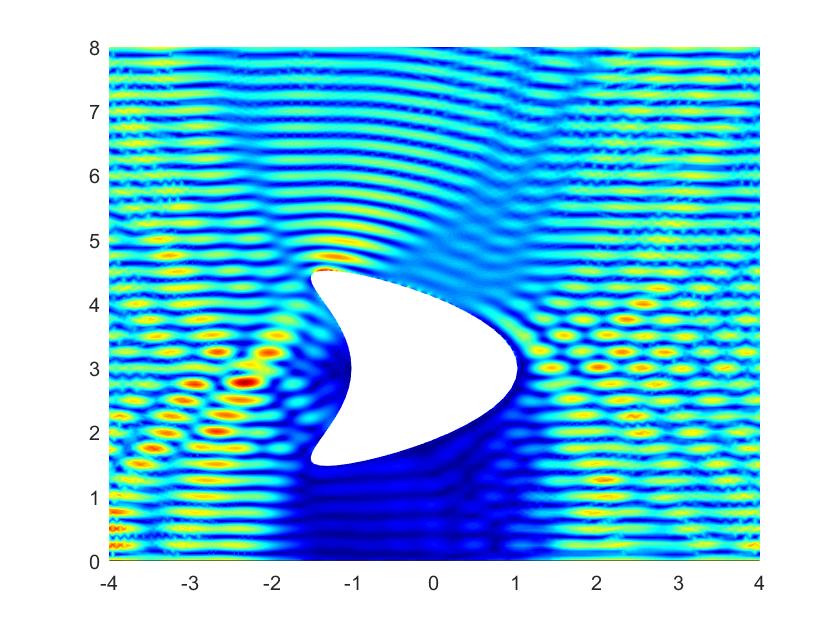} &
\includegraphics[scale=0.15]{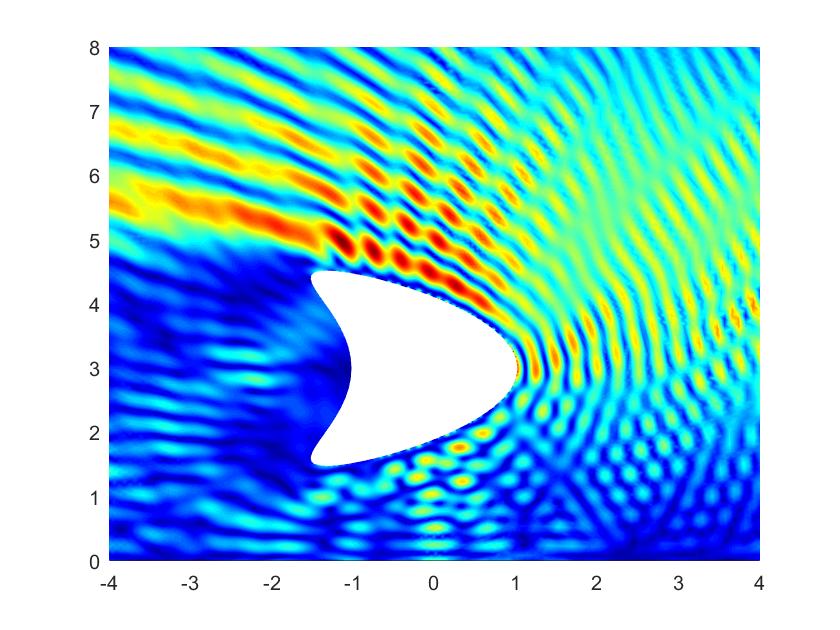} \\
(a) $\phi_s^{\mathrm{tot}}$ & (b) $\phi_s^{\mathrm{tot}}$ & (c) $\phi_s^{\mathrm{tot}}$ \\
\end{tabular}
\caption{Absolute values of the total fields $\phi_p^{\mathrm{tot}}$ and $\phi_s^{\mathrm{tot}}$ resulting from the HD-WGF method for the Dirichlet problem of scattering by a kite-shaped obstacle where $\omega=4\pi$, $A=16\lambda_s$, $c_p=1$ and $c_s=\sqrt{2}$. Here, (a,d): $\theta_p^\mathrm{inc}=0, \theta_s^\mathrm{inc}=0$; (b,e): $\theta_p^\mathrm{inc}=\frac{63\pi}{128}, \theta_s^\mathrm{inc}=0$ and (c,f): $\theta_p^\mathrm{inc}=\frac{63\pi}{128}, \theta_s^\mathrm{inc}=-\frac{63\pi}{128}$.}
\label{Figure2D.2}
\end{figure}

{\bf Three-dimensional examples.}
We test next the performance of the HD-WGF method for the Dirichlet problem of scattering by a spherical and ellipsoidal obstacle within a half space. Letting $\omega=4\pi$, $c_p=1$ and $c_s=\sqrt{2}$, Table~\ref{Table3D.1} displays the numerical errors $\epsilon_\infty$ for different values of $A$, which demonstrate the high accuracy and fast convergence of the proposed solver for all incidence angles. Figures~\ref{Figure3D.1} presents the computed values of the total field for the Dirichlet problem of scattering by a spherical obstacle for different incident angles, respectively. Finally, we consider the Dirichlet problem of scattering by a bean-shaped obstacle over a half space with $\omega=4\pi$, $c_p=1$, $c_s=0$, $\theta_p^\mathrm{inc}=-\frac{\pi}{3}$ and $A=5\lambda_s+1$. The real parts of the total fields $\phi_p^{\mathrm tot}$ and $\bs\phi_s^{\mathrm tot}$ are depicted in Figures~\ref{Figure3D.2}.

\begin{table}[htbp]
	\centering
	\begin{tabular}{|c|c|c|c|c|c|c|}
		\hline
		\multicolumn{7}{|c|}{Spherical obstacle} \\
		\hline
		\multirow{2}{*}{$A/\lambda_s$} & \multicolumn{3}{|c|}{$\theta_s^\mathrm{inc}=0$}  & \multicolumn{3}{|c|}{$\theta_s^\mathrm{inc}=-\frac{63\pi}{128}$}\\
		\cline{2-7}
		& $\theta_p^\mathrm{inc}=0$ & $\theta_p^{inc}=\frac{\pi}{4}$ & $\theta_p^\mathrm{inc}=\frac{63\pi}{128}$
		& $\theta_p^\mathrm{inc}=0$ & $\theta_p^\mathrm{inc}=\frac{\pi}{4}$ & $\theta_p^\mathrm{inc}=\frac{63\pi}{128}$  \\
		\hline
		2  & 3.05E-1 & 9.10E-2 & 1.34E-1 & 2.73E-1 & 9.98E-1 & 1.38E-1 \\
		4  & 2.85E-2 & 1.59E-2 & 4.97E-3 & 1.98E-2 & 3.48E-2 & 1.52E-2 \\
		6  & 6.94E-3 & 8.16E-3 & 5.05E-4 & 5.80E-3 & 4.83E-3 & 5.93E-3 \\
		8  & 4.34E-4 & 7.13E-4 & 4.87E-4 & 7.46E-4 & 1.02E-3 & 6.91E-4 \\
		\hline
		\multicolumn{7}{|c|}{Ellipsoidal obstacle} \\
		\hline
		\multirow{2}{*}{$A/\lambda_s$} & \multicolumn{3}{|c|}{$\theta_s^\mathrm{inc}=0$}  & \multicolumn{3}{|c|}{$\theta_s^\mathrm{inc}=-\frac{63\pi}{128}$}\\
		\cline{2-7}
		& $\theta_p^\mathrm{inc}=0$ & $\theta_p^\mathrm{inc}=\frac{\pi}{4}$ & $\theta_p^\mathrm{inc}=\frac{63\pi}{128}$
		& $\theta_p^\mathrm{inc}=0$ & $\theta_p^\mathrm{inc}=\frac{\pi}{4}$ & $\theta_p^\mathrm{inc}=\frac{63\pi}{128}$  \\
		\hline
		2  & 2.25E-2 & 2.07E-2 & 3.13E-2 & 3.86E-2 & 3.24E-1 & 1.82E-1 \\
		4  & 7.42E-3 & 9.48E-3 & 1.37E-2 & 6.51E-3 & 1.85E-2 & 9.54E-3 \\
		6  & 2.14E-3 & 5.40E-3 & 5.38E-3 & 1.51E-3 & 6.46E-3 & 6.80E-3 \\
		8  & 7.57E-4 & 2.04E-3 & 8.72E-4 & 5.86E-4 & 1.79E-3 & 1.57E-3 \\
		\hline
	\end{tabular}
	\caption{Relative errors $\epsilon_\infty$ in the total field resulting from the HD-WGF method for the problems of scattering by a spherical and an ellipsoidal obstacle on a half space where $c_p=1$ and $c_s=\sqrt{2}$.}
	\label{Table3D.1}
\end{table}

\begin{figure}[htbp]
	\centering
	\begin{tabular}{ccc}
		\includegraphics[scale=0.15]{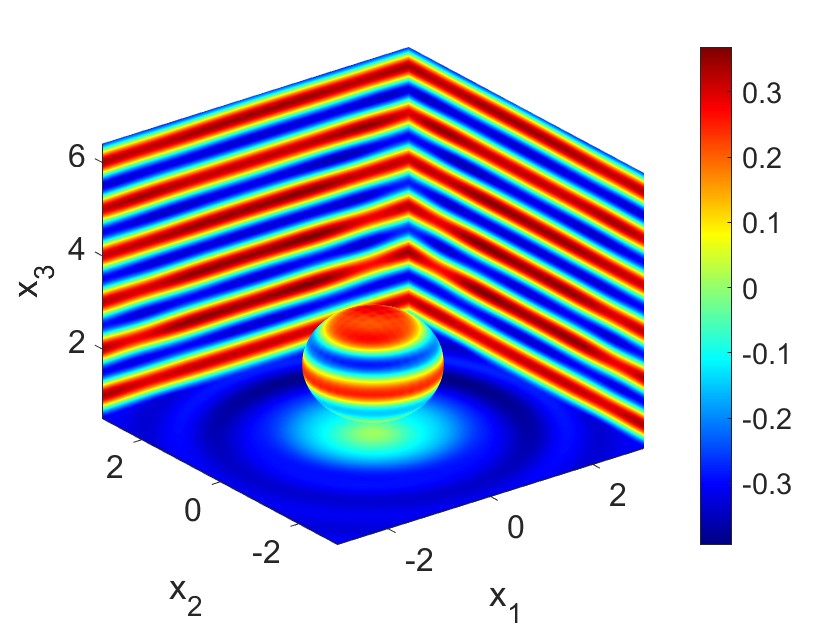} &
		\includegraphics[scale=0.15]{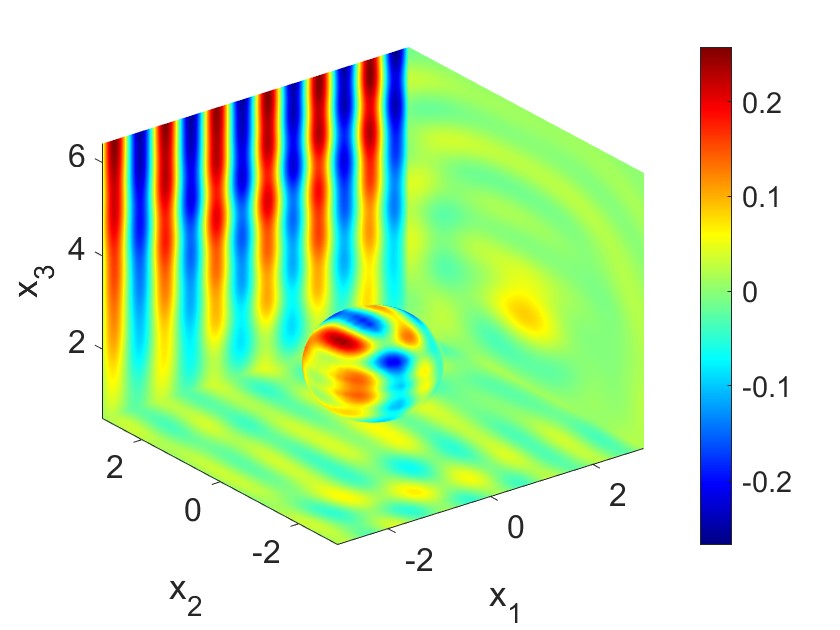} &
		\includegraphics[scale=0.15]{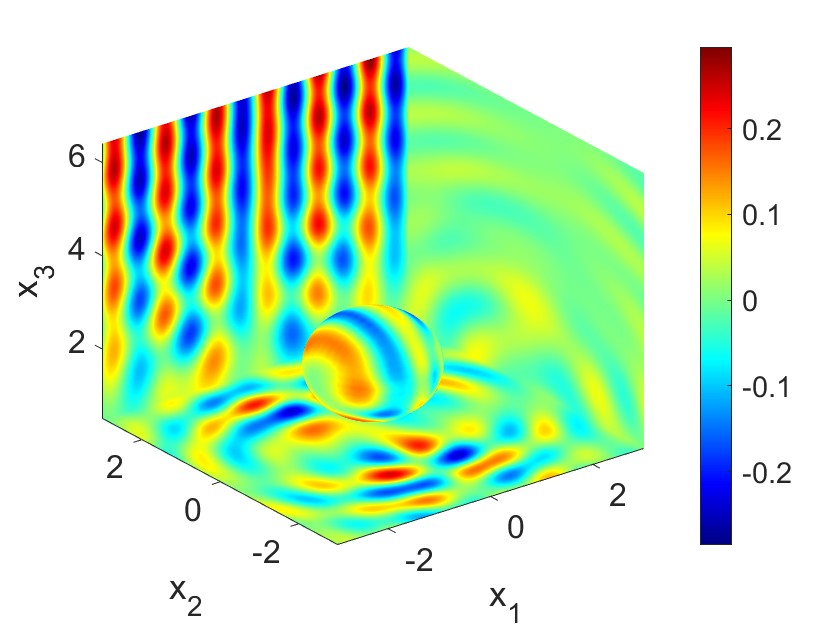} \\
		(a) Real($\phi_p^{\mathrm{tot}}$) & (b) Real($\phi_p^{\mathrm{tot}}$) & (c) Real($\phi_p^{\mathrm{tot}}$) \\
		\includegraphics[scale=0.15]{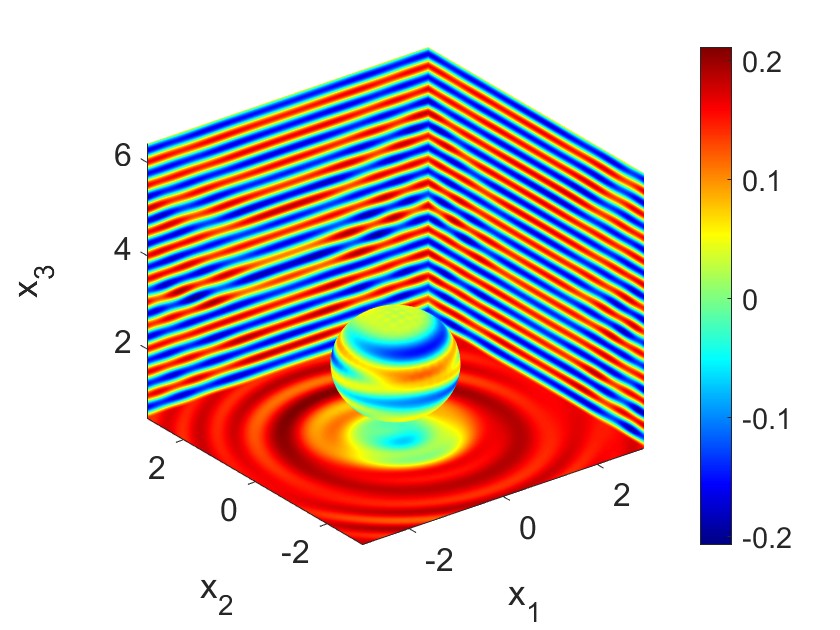} &
		\includegraphics[scale=0.15]{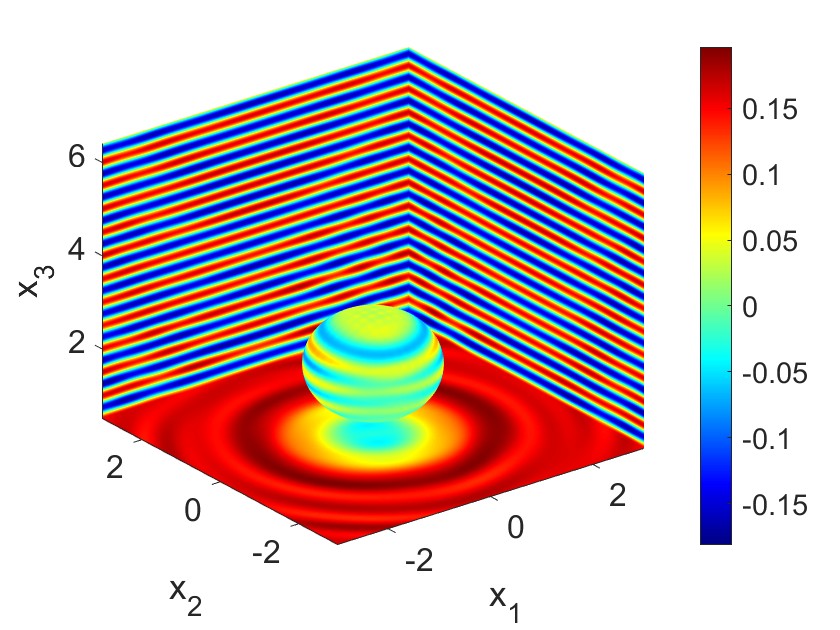} &
		\includegraphics[scale=0.15]{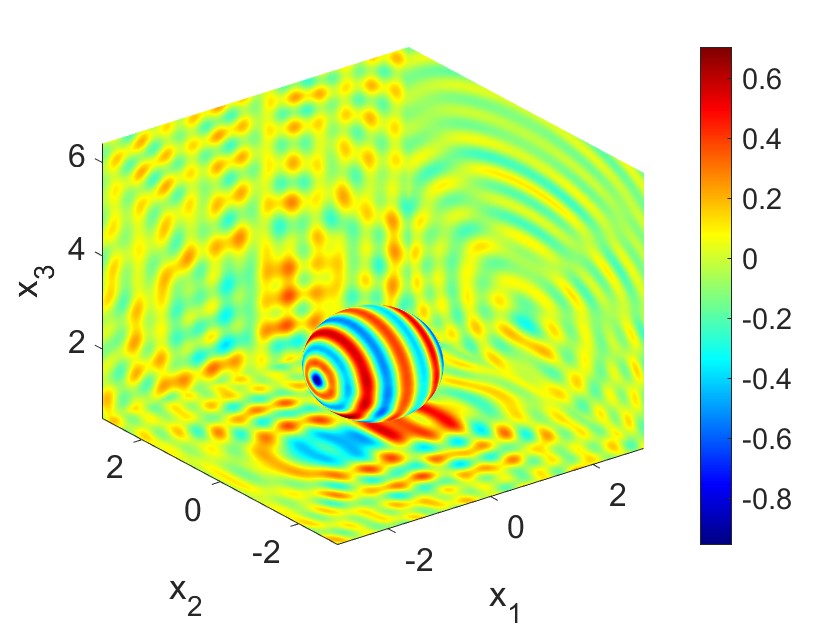} \\
		(d) Real($(\bs\phi_s^{\mathrm{tot}})_1$) & (e) Real($(\bs\phi_s^{\mathrm{tot}})_1$) & (f) Real($(\bs\phi_s^{\mathrm{tot}})_1$) \\
	\end{tabular}
	\caption{Real parts of the total fields $\phi_p^{\mathrm{tot}}$ and the first component of $\bs\phi_s^{\mathrm{tot}}$ resulting from the HD-WGF method for the Dirichlet problem of scattering by a spherical obstacle where $\omega=4\pi$, $A=5\lambda_s+1$, $c_p=1$ and $c_s=\sqrt{2}$. Here, (a,d): $\theta_p^\mathrm{inc}=0, \theta_s^\mathrm{inc}=0$; (b,e): $\theta_p^\mathrm{inc}=\frac{63\pi}{128}, \theta_s^\mathrm{inc}=0$ and (c,f): $\theta_p^\mathrm{inc}=\frac{63\pi}{128}, \theta_s^\mathrm{inc}=-\frac{63\pi}{128}$.}
	\label{Figure3D.1}
\end{figure}

\begin{figure}[htbp]
	\centering
	\begin{tabular}{cc}
		\includegraphics[scale=0.15]{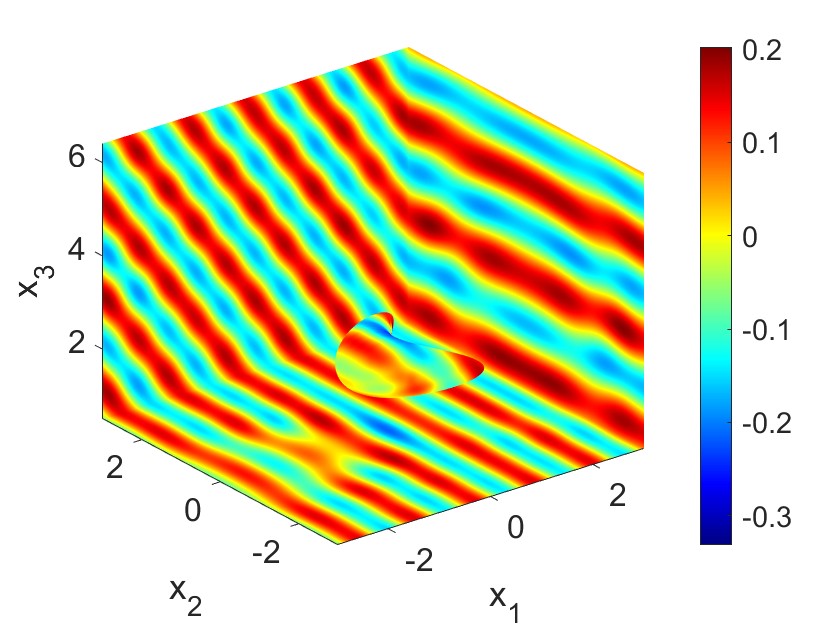} &
		\includegraphics[scale=0.15]{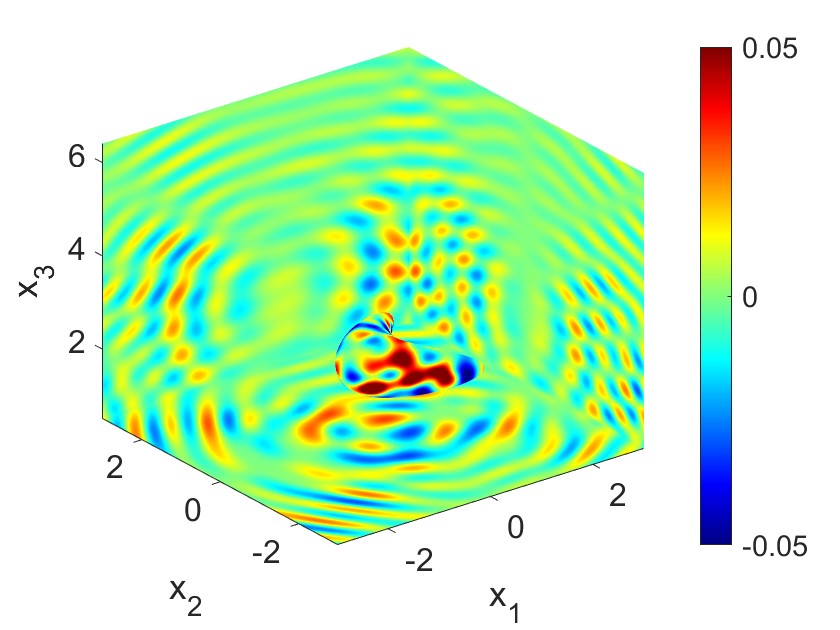} \\
		(a) Real($\phi_p^{\mathrm{tot}}$) & (b) Real($(\bs\phi_s^{\mathrm{tot}})_1$) \\
		\includegraphics[scale=0.15]{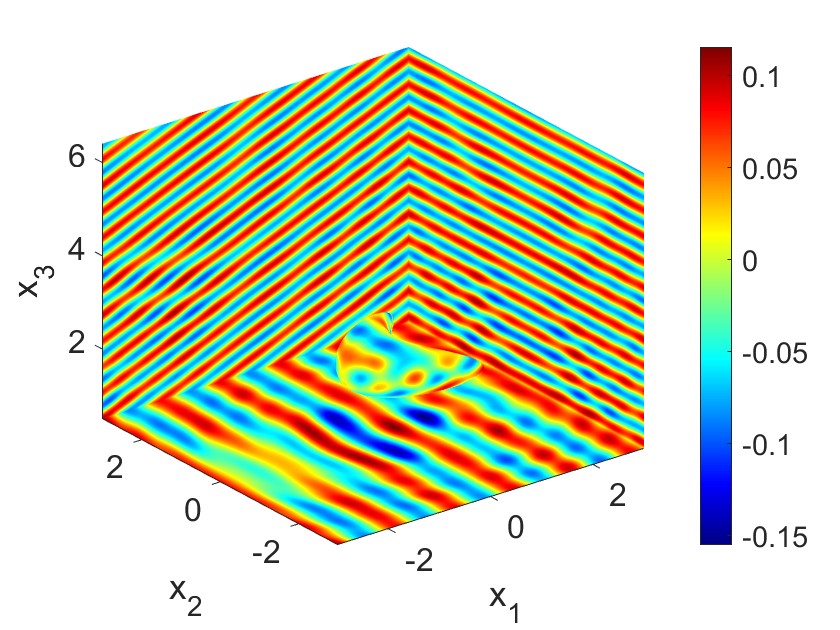} &
		\includegraphics[scale=0.15]{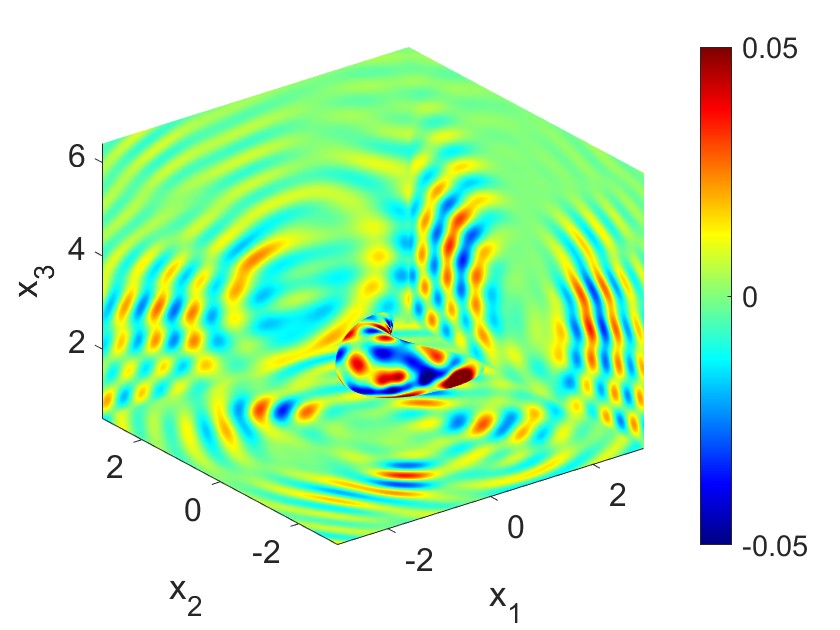} \\
		(c)Real($(\bs\phi_s^{\mathrm{tot}})_2$) & (d)Real($(\bs\phi_s^{\mathrm{tot}})_3$) \\
	\end{tabular}
	\caption{Real parts of the total fields $\phi_p^{\mathrm{tot}}$ and $\bs\phi_s^{\mathrm{tot}}$ resulting from the HD-WGF method for the Dirichlet problem of scattering by a bean-shaped obstacle where $\omega=4\pi$, $A=5\lambda_s+1$, $\theta_p^{\mathrm{inc}}=-\frac{\pi}{3}$, $c_p=1$ and $c_s=0$.}
	\label{Figure3D.2}
\end{figure}

\section{Conclusion}
\label{sec:6}

This paper proposed a novel HD-WGF method for solving the elastic scattering problems on a half-space in both 2D and 3D. By means of the Helmholtz decomposition, the original problem was transformed into a coupled system of pressure and shear waves, which satisfy the Helmholtz and Helmholtz/Maxwell equations, respectively, and the corresponding BIEs in terms of free-space fundamental solution of Helmholtz equation were derived. Then the windowing idea was introduced for the purpose of truncation and a ``correction'' strategy was proposed, as demonstrated by a variety of numerical
tests, to ensure the uniformly fast convergence for all incident angles of plane incidence. The extensions of the HD-WGF method to the elastic layered-medium problems with Neumann or transmission boundary conditions and multi-physics problems, including e.g. fluid-solid interaction problems, electromagnetic-elastic coupled problems, etc., are left for future work.

\section*{Acknowledgments}
TY gratefully acknowledges support from China NSF Grants 12171465 and 12288201. WZ was supported in part by China NSF grant 12226354 and the China NSF for Distinguished Young Scholars 11725106. 

\appendix
\section{Exact solutions to the problems of scattering by infinite plane}
\label{sec:A}
\renewcommand{\theequation}{A.\arabic{equation}}
\renewcommand{\thetheorem}{A.\arabic{theorem}}

The expressions of the exact solutions $\phi_{p,\Pi}^\mathrm{tot}, \phi_{s,\Pi}^\mathrm{tot}/ \boldsymbol \phi_{s,\Pi}^\mathrm{tot}$ to the problems of scattering by infinite plane $\Pi$ are given as follows.
\begin{itemize}
\item {\bf Two-dimensional case.} 
If the incident field ${\boldsymbol u}^\mathrm{inc}$ only contains the plane pressure wave ${\boldsymbol u}_p^\mathrm{inc}$ (see \eqref{2Dinc}), then the total field is given by ${\boldsymbol u}_{p,\Pi}^\mathrm{tot}={\boldsymbol u}_p^\mathrm{inc}+{\boldsymbol u}_p^\mathrm{ref}$, where 
$\bs u^\mathrm{ref}_p$ admits the Helmholtz decomposition
\[
\bs u^\mathrm{ref}_p=\nabla \phi_{p,p}^\mathrm{ref}+ \overrightarrow{\mathrm{curl}}\,\phi_{s,p}^\mathrm{ref},
\]
with
\[
\phi_{p,p}^\mathrm{ref}=\frac{l^p_p}{k_p}e^{i (\alpha_px_1+\beta_{p,p}x_2)},\quad \phi_{s,p}^\mathrm{ref}=\frac{l^p_s}{k_s}e^{i (\alpha_px_1+\beta_{s,p}x_2)}.
\]
Therefore, together with
\ben
\phi_p^\mathrm{inc}=\frac{c_p}{k_p}e^{i(\alpha_px_1-\beta_{p,p}x_2)},\quad \phi_s^\mathrm{inc}=0,
\enn
we can get $\phi_{\xi,p,\Pi}^\mathrm{tot}=\phi_{\xi}^\mathrm{inc}+\phi_{\xi,p}^\mathrm{ref}$, $\xi=p,s$. To determine the unknown parameters $l_\xi^p, \xi=p,s$, using the Dirichlet boundary condition on $\Pi$, the following linear system results:
\begin{eqnarray}
	\begin{bmatrix}
	\alpha_p & \beta_{s,p} \\
	\beta_{p,p} & -\alpha_p
	\end{bmatrix}\begin{bmatrix}
	l^p_p\\
	l^p_s
	\end{bmatrix} =\begin{bmatrix}
	-\alpha_p \\
	\beta_{p,p}
	\end{bmatrix},
	\end{eqnarray}
which, by directly calculating, further gives
\ben
l^p_p=\frac{\beta_{p,p}\beta_{s,p}-\alpha_p^2}{\beta_{p,p}\beta_{s,p}+\alpha_p^2},\quad l^p_s=-\frac{2\alpha_p\beta_{p,p}}{\beta_{p,p}\beta_{s,p}+\alpha_p^2}.
\enn

If the incident field ${\boldsymbol u}^\mathrm{inc}$ only contains the plane shear wave ${\boldsymbol u}_s^\mathrm{inc}$ (see \eqref{2Dinc}), then the total field is given by ${\boldsymbol u}_{s,\Pi}^\mathrm{tot}={\boldsymbol u}_s^\mathrm{inc}+{\boldsymbol u}_s^\mathrm{ref}$, where 
$\bs u^\mathrm{ref}_s$ admits the Helmholtz decomposition
\[
\bs u^\mathrm{ref}_s=\nabla \phi_{p,s}^\mathrm{ref}+ \overrightarrow{\mathrm{curl}}\,\phi_{s,s}^\mathrm{ref},
\]
with
\[
\phi_{p,s}^\mathrm{ref}=\frac{l^s_p}{k_p}e^{i (\alpha_sx_1+\beta_{p,s}x_2)},\quad \phi_{s,s}^\mathrm{ref}=\frac{l^s_s}{k_s}e^{i (\alpha_sx_1+\beta_{s,s}x_2)}.
\]
Therefore, together with
\ben
\phi_p^\mathrm{inc}=0,\quad \phi_s^\mathrm{inc}=\frac{c_s}{k_s}e^{i(\alpha_sx_1-\beta_{s,s}x_2)},
\enn
we can get $\phi_{\xi,s,\Pi}^\mathrm{tot}=\phi_\xi^\mathrm{inc}+\phi_{\xi,s}^\mathrm{ref}$, $\xi=p,s$. To determine the unknown parameters $l_\xi^s, \xi=p,s$, using the Dirichlet boundary condition on $\Pi$, the following linear system results:
\begin{eqnarray}
	\begin{bmatrix}
	\alpha_s & \beta_{s,s}\\
	\beta_{p,s} & -\alpha_s
	\end{bmatrix}\begin{bmatrix}
	l^s_p\\
	l^s_s
	\end{bmatrix} =\begin{bmatrix}
	\beta_{s,s} \\
	\alpha_s
	\end{bmatrix},
	\end{eqnarray}
which, by directly calculating, further gives
\ben
l^s_p=\frac{2\alpha_s\beta_{s,s}}{\beta_{p,s}\beta_{s,s}+\alpha_s^2},\quad l^s_s=\frac{\beta_{p,s}\beta_{s,s}-\alpha_s^2}{\beta_{p,s}\beta_{s,s}+\alpha_s^2}.
\enn
Therefore, if ${\bs u}^\mathrm{inc}= c_p{\bs u}_p^\mathrm{inc}+c_s{\bs u}_s^\mathrm{inc}$, where $c_\xi,\xi=p,s$ are constants, we have $\phi_{\xi,\Pi}^\mathrm{tot}=\phi_{\xi,p,\Pi}^\mathrm{tot}+\phi_{\xi,s,\Pi}^\mathrm{tot}$, $\xi=p,s$.

\item {\bf Three-dimensional case.} Let's onsider next the problem of scattering by a flat surface in three dimensions. Assume that ${\boldsymbol u}^\mathrm{inc}$ is a compressional plane wave field ${\boldsymbol u}_p^\mathrm{inc}$ (see \eqref{3Dinc}). Analogous to the derivation for the two-dimensional problem, it follows that
\begin{align}
& \phi_{p,p}^\mathrm{ref}=
-\frac{(\boldsymbol \alpha_p^\top, \beta_{s,p})^\top\cdot (\boldsymbol \alpha_p^\top, -\beta_{p,p})^\top}{k_p(\beta_{p,p}\beta_{s,p}+|\boldsymbol\alpha_p|^2)}e^{i (\boldsymbol\alpha_p\cdot \widetilde{\boldsymbol x}+\beta_{p,p}x_3)},\\
&\bs\phi_{s,p}^\mathrm{ref}=-\frac{(\boldsymbol \alpha_p^\top, -\beta_{p,p})^\top\times (\boldsymbol \alpha_p^\top, \beta_{p,p})^\top}{k_p(\beta_{p,p}\beta_{s,p}+|\boldsymbol\alpha_p|^2)} e^{i (\boldsymbol\alpha_p\cdot \widetilde{\boldsymbol x}+\beta_{s,p}x_3)}.
\end{align}
Considering the shear incident plane wave field ${\boldsymbol u}_s^\mathrm{inc}$, it follows that
\begin{align}
&\phi_{p,s}^\mathrm{ref}=
-\frac{(\boldsymbol \alpha_s^\top, \beta_{s,s})^\top\cdot\boldsymbol d^\perp}{\beta_{p,s}\beta_{s,s}+|\boldsymbol\alpha_s|^2}e^{i (\boldsymbol\alpha_s\cdot \widetilde{\boldsymbol x}+\beta_{p,s}x_3)},\\
&\bs\phi_{s,s}^\mathrm{ref}=-\frac{\boldsymbol d^\perp\times (\boldsymbol \alpha_s^\top, \beta_{p,s})^\top}{\beta_{p,s}\beta_{s,s}+|\boldsymbol\alpha_s|^2}e^{i (\boldsymbol\alpha_s\cdot \widetilde{\boldsymbol x}+\beta_{s,s}x_3)},
\end{align}
where $\bs d^\perp:=\frac{1}{k_s}\begin{pmatrix}
	\bs \alpha_s^\top\\
	-\beta_{s,s}
	\end{pmatrix}\times \bs d_s$. Therefore, for ${\bs u}^\mathrm{inc}= c_p{\bs u}_p^\mathrm{inc}+c_s{\bs u}_s^\mathrm{inc}$, where $c_\xi,\xi=p,s$ are constants, the exact solutions $\phi_{p,\Pi}^\mathrm{tot}$ and $ \bs\phi_{s,\Pi}^\mathrm{tot}$ can be calculated through $\phi_{p,\Pi}^\mathrm{tot}=\phi_{p,p,\Pi}^\mathrm{tot}+\phi_{p,s,\Pi}^\mathrm{tot}$ and $\bs\phi_{s,\Pi}^\mathrm{tot}=\bs\phi_{s,p,\Pi}^\mathrm{tot}+\bs\phi_{s,s,\Pi}^\mathrm{tot}$. 

\end{itemize}

\end{document}